\documentclass{article}

\usepackage[utf8]{inputenc}
\usepackage{dsfont}
\usepackage{mathtools}
\usepackage{amsthm}
\usepackage{amsmath}
\usepackage{amssymb}
\usepackage{mathrsfs}
\usepackage{lmodern}
\usepackage{ae}
\usepackage{enumerate}
\usepackage{comment}
\usepackage{stmaryrd}
\usepackage{hyperref}
\hypersetup{colorlinks=true}
\newcommand\fnsep{\textsuperscript{\, ,}}

\renewcommand{\Re}{\mathrm{Re}}
\newcommand{\af}{\mathbf{b}}
\newcommand{\cf}{\mathbf{b}^*}
\newcommand{\ab}{\mathbf{a}}
\newcommand{\cb}{\mathbf{a}^*}
\newcommand{\wf}{\omega^{(\mathrm{f})}}
\newcommand{\wb}{\omega^{(\mathrm{b})}}

\newcommand{\mf}{m_\mathrm{f}}
\newcommand{\mb}{m_\mathrm{b}}

\newcommand{\nop}{n_\mathrm{Op}}
\newcommand{\tnop}{\tilde{n}_{\mathrm{Op}}}

\usepackage[left=1.5cm,right=1.5cm,top=2cm,bottom=2cm]{geometry}

\newtheorem{rk}{Remark}
\newtheorem{lem}{Lemma}[section]

\newtheorem{Thannex}{Theorem}[section]
\newtheorem{Th}{Theorem}

\newtheorem{Hypothesis}{Hypothesis}

\title{Ultraviolet Renormalisation of a quantum field toy model I }
\author{Benjamin Alvarez\footnote{Corresponding author}\fnsep\footnote{Email: \texttt{benjamin.alvarez@univ-tln.fr}} \\ Aix Marseille Univ, Univ Toulon, CNRS, CPT, Marseille, France \\
\\
Jacob Schach M\o ller\footnote{Email: \texttt{jacob@math.au.dk}} \\ Department of Mathematics, Aarhus University, Denmark}

\begin{document}

\maketitle

\begin{center}
\textbf{Abstract}
\end{center}

We consider a class of translation invariant Hamiltonians describing a fermion field coupled to a boson field.  The interaction kernels are assumed bounded in the fermionic momentum variable and decaying like $|q|^{-p}$ for large boson momenta $q$. A realistic physical case would be $p=\frac12$. We impose a spatial cutoff and UV-renormalise the resulting Hamiltonian by subtracting its ground state energy. Our renormalisation procedure works for the physically realistic value $p=\frac12$ in spatial dimensions $1$ and $2$ and for $p>\frac34$ in spatial dimension $d=3$.  

\section{Introduction} 

Ultraviolet (UV) Renormalisation is a well known and important problem in quantum field theory. We refer to \cite{Das:2008zze,Peskin:1995ev} for an introduction to this topic from a physical point of view. Even if this procedure leads to predictions in agreement with experiments, it remains a difficult mathematical problem, which is necessary to tackle if one wants to define a physical model without cutoffs. Nowadays, UV renormalisation goes through a Lagrangian formulation, with a renormalised Hamiltonian arising indirectly as a generator of time translations in a representation of the Poincare group. In  this paper we are concerned with UV renormalisation at the level of Hamiltonians. Among other works, we refer to \cite{GlJa77_01} where both ultraviolet and spatial cutoffs have been removed for the Yukawa model and for $\Phi_2^{2n}$ in $1+1$ dimension. In particular, \cite{GlJa77_01} the Yukawa model is renormalised through a shift in self-energy corresponding to the ground state energy to second order in perturbation theory, together with a second order boson mass shift. Their method to remove the ultraviolet cutoffs is different from ours and relies on $N_{\tau}$ estimates. Here we are concerned with the removal of the ultraviolet cutoff, not with the spatial cutoff, which in \cite{GlJa77_01} is removed using abstract arguments from the theory of automorphisms of operator algebras. 

The method we employ has its origin in works of Eckmann and Hepp \cite{Eck1970,Hepp1969}. See also \cite{Fro1974} for an adaptation of the results of \cite{Eck1970} in the context of the relativistic Nelson model, where the Gross dressing transformation \cite{GriWun2018,Nel1964} does not apply. Our starting point is a more recent adaption from the PhD thesis of W\"unsch \cite{AW2017}, where the strategy of Eckmann was carried out in detail for general models of $N$ fermions interacting with a bosonic scalar field through a linear interaction preserving the fermion number. See also \cite{BrDe07,De03_01,Gross1973,GuJoHi2014,Lampart2020,MaMo2018,Schmidt2019,Sloan1974,TeTu2021} for other self-energy renormalisation methods at the level of Hamiltonians, all pertaining to models with conserved  fermion number, or without fermions present.

The quantum field toy model we propose to study involves a scalar boson field in interaction with a scalar fermion field. 
The main difference when comparing with the other models renormalised through Eckmann's method, is that our toy model has no conserved particle number. This toy model has no physical relevance (fermions being scalar), but we justify it in the following way.  Eckmann's strategy \cite{Eck1970} provides a powerful tool combined with conserved particle number or where only one field is involved \cite{AW2017}. However, most physical models, involves products of at least three creation and/or annihilation operators in their interaction term without relevant conservation of the number of particles. The interaction term in our model involves products of fermionic and bosonic annihilation and creation operators with momentum conservation, thus filling a gap in complexity between the linearly coupled models with conserved fermion number and the Yukawa model studied in \cite{GlJa77_01}. As in the latter physical models, we impose a spatial cutoff that breaks translation invariance.

The starting point for Eckmann's method is a suitable resummation of the Neumann series of the resolvent, grouping singular expressions with counterterms.  When only one particle is involved (or equivalently when the number of particles is conserved), the singular expressions involve only finite sequences of creation and annihilation operators. The counterterms are designed to regularise the fully contracted term that comes out after normal ordering a singular expression. If one includes an additional bosonic field, as in the minimally coupled model, then applying the same strategy would require to normal order completely every term in the Neumann series leading to sequences of factorial growth. The situation is simpler, however, if one considers -- as we do -- a fermion field interacting with a boson field instead, as the boundedness of smeared fermionic annihilation and creation operators may be exploited, something that also plays a crucial role in \cite{GlJa77_01}.

The present paper is a first step, only taking into account leading order contributions to the self-energy. In a follow-up paper \cite{AlMo22}, we systematically take into account higher-order contributions to the self-energy that require the identification and normal ordering, of longer and longer sequences of terms in a Neumann expansion. This will take us to the limit of what one may expect, by comparison with solvable models \cite{De03_01}. Note that first results in this direction have already been obtained, with a different method, in the context of the Polaron model \cite{Lampart2022}.  Finally, in future work, the authors hope to tackle models with an interaction quadratic in bosonic fields. All these projects are dedicated to extend the applicability of Eckmann's method, in order to apply it to more complex physical models, like the Yukawa model. 

We will now introduce, in a more precise way, the toy model. It is built by analogy to the general expression of quantum field interaction terms. The interaction term is quadratic and of the same form as $\int \Psi_\mathrm{b}(x) \Psi_\mathrm{f}(x) dx$, where $\Psi_\mathrm{b}(x)$ is a scalar boson field and $\Psi_\mathrm{f}$ is a scalar fermion field. For the boson field, we may use the relativistic free field
\begin{equation}\label{FieldBoson}
\Psi_\mathrm{b}(x) = (2 \pi)^{-\frac{d}{2}} \int_{\mathbb{R}^d} \biggl[ \frac{ e^{i q \cdot x}}{(2(|q|^2+\mb^2)^{\frac12})^{\frac12}}\ab(q) +  \frac{ e^{-i q \cdot x}}{(2(|q|^2+\mb^2)^{\frac12})^{\frac12}}\cb(q) \biggr]d q.
\end{equation}
We will be working with massive bosons $\mb>0$, although we believe that one may also be able treat massless bosons using the methods of this paper.
There are no scalar fermions, so we are instead guided by the form of the relativistic free spin-$\frac12$ fermion field, which is of them form
\[
 (2 \pi)^{-\frac{d}{2}} \sum_{s\in \{-\frac12,\frac12\} } \int_{\mathbb{R}^d} \biggl[ \frac{u(k,s) e^{i k \cdot x}}{(2(|k|^2+\mf^2)^{\frac12})^{\frac12}}\af_{+}(k,s) +  \frac{v(k,s) e^{-i k \cdot x}}{(2(|k|^2+\mf^2)^{\frac12})^{\frac12}}\cf_{-}(k,s) \biggr]d k,
\]
where the operators $\af_{+}(k,s)$ are annihilation operators for the fermion and the operators $\mathrm{\bf b}_{-}^*(k,s)$ are creation operators for the anti fermion. Since we will be working with scalar fermions, these operator will not appear again.  The second conjugate field is of the same form.
The functions $u$ and $v$ come from Dirac spinors and
\begin{equation}\label{PhysCouplings}
k \mapsto \frac{u(k,s)}{(2(|k|^2+\mf^2)^{\frac12})^{\frac12}} \qquad
\textup{and} \qquad  
k\mapsto \frac{v(k,s) }{(2(|k|^2+\mf^2)^{\frac12})^{\frac12}}
\end{equation}
are bounded functions. This leads us to consider a scalar fermion field of them form
\begin{equation}\label{FieldFermion}
\Psi_\mathrm{f}(x) =  (2 \pi)^{-\frac{d}{2}} \int_{\mathbb{R}^d} \Bigl[ \overline{h(k)} e^{i k \cdot x}\af(k) +  h(k) e^{-i k \cdot x}\cf(k) \Bigr]d k
\end{equation}   
with $h\colon \mathbb R^d\to \mathbb C$ a bounded (measurable) function. We will be using the letter $k$ for fermion momentum throughout the paper, whereas the letter $q$ will denote boson momentum.

The two particle toy model is defined on  a tensor product of a bosonic Fock space with a fermionic Fock space: 
\[
\mathscr{H} = \mathscr{F}_{\mathrm{s}}\bigl(L^2(\mathbb R^d)\bigr) \otimes \mathscr{F}_{\mathrm{a}}\bigl(L^2(\mathbb R^d)\bigr),
\]
where the subscript s stands for symmetric, and a for antisymmetric.
The free Hamiltonian $H_0$ is defined as the sum of the second quantisations of the dispersion relations for each particle:
\[
H_0 = \mathrm{d}{\Gamma}\bigl(\wb\bigr) + \mathrm{d}{\Gamma}\bigl(\wf\bigr),
\]
or equivalently,
\[
H_0 = \int_{\mathbb{R}^d} \wf(k) \cf(k)\af(k)dk + \int_{\mathbb{R}^d}  \wb(q) \cb(q)\ab(q)dq,
\]
where $\wb$ is the dispersion relation for the boson and $\wf$ is the one for the fermions. They are furthermore of the following form:  
\[
\wb(q)= \sqrt{q^2 + \mb^2},\qquad \wf(k)= \sqrt{k^2 + \mf^2} \qquad\text{with }m_{\mathrm{b}},\,m_{\mathrm{f}} > 0.
\] 
The operator $\cf$, respectively $\af$, stands for the creation, respectively annihilation, operator for fermions. The operator $\cb$, respectively $\ab$, stands for the creation, respectively annihilation, operator for bosons. The standard canonical commutation and anticommutation relations hold:
\[
\begin{aligned}
[\cb(q),\cb(q')] & =  [\ab(q),\ab(q')] = 0, & \qquad 
[\cb(q),\cf(k)] & =  [\cb(q),\af(k)] = 0,\\
[\ab(q),\cf(k)] & =  [\ab(q),\af(k)] = 0, & \qquad 
 [\ab(q),\cb(q')] & =  \delta(q-q')\\
\{\cf(k), \cf(k') \} & =   \{\af(k), \af(k') \} = 0, & \qquad 
\{\cf(k), \af(k') \} & =  \delta(k-k').
\end{aligned}
\]
Here $[A,B] = AB-BA$ and $\{A,B\} = AB+BA$.
Moreover, $H_0$ is essentially selfadjoint on:
\[
\mathscr H_\mathrm{fin}=\mathscr{F}_{\mathrm{s,fin}}\bigl(C^{\infty}_0(\mathbb{R}^d)\bigr) \otimes \mathscr{F}_{\mathrm{a,fin}}\bigl(C^{\infty}_0(\mathbb{R}^d)\bigr),
\]
where $\mathscr{F}_{\mathrm{s,fin}}(C^{\infty}_0(\mathbb{R}^d))$, respectively $\mathscr{F}_{\mathrm{a,fin}}(C^{\infty}_0(\mathbb{R}^d))$, is the subspace of $\mathscr{F}_\mathrm{s}$, respectively $\mathscr{F}_\mathrm{a}$, of states $(\phi_0, \phi_1, \dots) \in \mathscr{F}_{\sharp}$ such that  $\phi_l\in \otimes^l_{k=1} C^{\infty}_0(\mathbb{R}^d)$ and the functions $\phi_l$ are vanishing but for finitely many $l$, and the tensor products here are the algebraic tensor products. 

\newcommand{\Ph}{\mathrm{Phys}}

We will consider an interaction of the form
\begin{equation}
\label{InteractionIntro}
H_\mathrm{I}\bigl(G^{(1)}_\Ph,G^{(2)}_\Ph\bigr) = H^{\cb \af}\bigl(G^{(1)}_\Ph\bigr)+H^{ \ab \cf }\bigl(G^{(1)}_\Ph\bigr)+
H^{\ab \af}\bigl(G^{(2)}_\Ph\bigr) + H^{\cb \cf}\bigl(G^{(2)}_\Ph\bigr),
\end{equation}
where 
\[
\begin{aligned}
H^{\cb \af}(G^{(1)}_\Ph)& =  \int \overline{G^{(1)}_\Ph(k,q)} \af(k) \cb(q) dk dq, & H^{ \ab \cf }(G^{(1)}_\Ph) & =  \int G^{(1)}_\Ph(k,q) \cf(k) \ab(q) dk dq,\\
H^{\ab \af}(G^{(2)}_\Ph) & =  \int \overline{G^{(2)}_\Ph(k,q)} \af(k) \ab(q) dk dq, & H^{\cb \cf}(G^{(2)}_\Ph) &=  \int G^{(2)}_\Ph(k,q)\cf(k) \cb(q) dk dq.
\end{aligned}
\]
The interaction $\int \Psi_\mathrm{b}(x) \Psi_\mathrm{f}(x) dx$, with fields given by \eqref{FieldBoson} and \eqref{FieldFermion}, will take the above form with $G^{(1)}(k,q) = G^{(2)}(k,q) = \frac{h(k)}{\wb(q)^\frac12}$.
 The physically motivated kernels are therefore assumed to be of the form: 
\begin{equation}
\label{condiGPh}
G^{(1)}_\Ph(k,q)  =  \frac{h^{(1)}(k,q) }{\wb(q)^\frac12}\delta(k-q),\qquad
G^{(2)}_\Ph(k,q)  =   \frac{h^{(2)}(k,q) }{\wb(q)^\frac12} \delta(k+q),
\end{equation}
where $\delta$ denotes the delta function, which ensures momentum conservation. 
Motivated by the properties of the functions \eqref{PhysCouplings},
we impose the following condition on the functions $h^{(1)}$ and $h^{(2)}$:
\begin{Hypothesis}
\label{Hypothesis-h}
For $j=1,2$,  $h^{(j)}\in L^\infty(\mathbb R^d\times\mathbb R^d)$.
\end{Hypothesis}
Note that $H_\mathrm{I}$, as well as its four terms, are well-defined as forms on $\mathscr H_\mathrm{fin}$. The unrenormalised interacting Hamiltonian, as a form on $\mathscr H_\mathrm{fin}$, is defined by the form sum
\begin{equation}\label{Model}
H= H\bigl(G^{(1)}_\Ph,G^{(2)}_\Ph\bigr) = H_0 + H_\mathrm{I}\bigl(G^{(1)}_\Ph,G^{(2)}_\Ph\bigr).
\end{equation}
In \cite{Al19}, one of us considered a simplified version of this toy model, with $G_\Ph^{(2)}=0$. The resulting model conserves total particle number and is much easier to handle. In fact, for small coupling, it does not require renormalisation.

 As discussed in the introduction, Hypothesis~\ref{Hypothesis-h} is not enough to define \eqref{Model} as a self-adjoint operator. We have to regularise the kernels in two ways. First 
by introducing an ultraviolet cutoff in the form of cutoff function $\chi$ and a real number $\Lambda>0$, setting the scale of the UV cutoff. The cutoff function should satisfy
\begin{Hypothesis}
\label{Hypothesis-chi}
The function $\chi\in L^\infty(\mathbb R^d)$ is real-valued, non-negative, and has compact support $\mathrm{supp}(\chi)$. We furthermore assume that $\chi$ is continuous at $0$ with $\chi(0)=1$. For $\Lambda>0$, we set
$\chi_\Lambda(k) = \chi(k/\Lambda)$.
\end{Hypothesis}
The second regularisation is a spatial cutoff $g$ that replaces the delta function and therefore breaks momentum conservation. We have
\begin{equation}
\label{RegularisedHamiltonian}
H_{\Lambda}= H\bigl(G_{\Lambda}^{(1)},G^{(2)}_\Lambda\bigr) = H_0 +  H_\mathrm{I}\bigl(G^{(1)}_{\Lambda},G^{(2)}_{\Lambda}\bigr),
\end{equation}
where, for $j=1,2$,
\begin{equation}\label{KernelApprox}
G^{(j)}_{\Lambda}(k,q) = G^{(j)}(k,q) \chi_\Lambda(k)\chi_\Lambda(q)
\end{equation}
and 
\begin{equation}
\label{condiG}
G^{(1)}(k,q)  =  \frac{h^{(1)}(k,q) }{\wb(q)^p}g(k-q),\qquad
G^{(2)}(k,q)  =   \frac{h^{(2)}(k,q) }{\wb(q)^p} g(k+q).
\end{equation}
The exponent $p$ is supposed to be equal to $\frac12$ in the physically motivated case. We suppress the dependence on $p$, $g$ and $\chi$ in the notation $G^{(j)}_{\Lambda}$ and $G^{(j)}$, $j=1,2$. 

 

\begin{Hypothesis}
\label{MainHypothesis}
The spatial cutoff $g\in L^\infty(\mathbb R^d)$ has compact support, $\mathrm{supp}(g)$, contained in the unit ball $\mathscr{B}(0,1) = \{z\in\mathbb C\, |\, |z|<1\}$.
\end{Hypothesis}

\begin{rk} In order to remove the spatial cutoff, one should take it to be of the form $g_n(k) = n^d f(n k)$ with $d$ the dimension and $n \in \mathbb{N}$. The function $f\in L^\infty(\mathbb R^d)$ should have support in the unit ball and  satisfy that $f\geq 0$
and
 \begin{equation}
 \int f(k) dk = 1.
 \end{equation}
In the sense of distributions, we would then have $\lim_{n\to \infty}g_n = \delta$. All the central constants appearing in the estimates of this paper are $L^2$-norms of functions involving $G_\Lambda^{(1)}$ and $G_\Lambda^{(2)}$. While these constants are all uniformly bounded in the UV scale $\Lambda$, they diverge when the spatial cutoff is removed, i.e., when $n\to \infty$.
\end{rk}

\begin{rk}
 We have not made any attempts to accommodate more general cutoff functions $\chi$ and $g$. The choices made here are surely not optimal. One could easily work with a larger class of cutoffs that do not have compact support. 
\end{rk}

 As we shall see in Theorem~\ref{ThSA} below, the UV-regularised Hamiltonian $H_{\Lambda}$ defines a self-adjoint operator, which is bounded from below.  We can now state our main result:
\begin{Th}
\label{MainTh}
Assume that the Hamiltonian $H_\Lambda$, $\Lambda>0$, is of the form \eqref{RegularisedHamiltonian} with an interaction kernel given by \eqref{KernelApprox}, where $h^{(1)}, h^{(2)}$ satisfy Hypothesis~\ref{Hypothesis-h}, $\chi$ satisfy Hypothesis~\ref{Hypothesis-chi}, and $g$ fulfils Hypothesis~\ref{MainHypothesis}. Let $E_{\Lambda} = \inf\sigma(H_{\Lambda})$. If moreover $p>\frac{d}{2}-\frac{3}{4}$, then there exists a self-adjoint operator $H$, such that $H_{\Lambda}-E_{\Lambda}$ converges in the norm resolvent sense to $H$. The renormalised Hamiltonian $H$ does not depend on the choice of $\chi$.
\end{Th} 

In spatial dimensions one and two, it is therefore possible to define a quantum field model without ultraviolet cutoffs but spatial cutoffs remain. We believe that this result is not optimal and we currently work on an improvement of Theorem~\ref{MainTh}.

This paper is organised as follows. The second part is dedicated to the definition of the regularised Hamiltonian and the introduction of the counterterm needed during the cancellations   of the divergence. The third part is dedicated to the proof of Theorem \ref{MainTh}.

Throughout this paper, we use the following convention 
for sums and products:  $ \sum^{k-1}_{i=k} a_i = \sum_{i\in\emptyset} a_i = 0$ and $\prod^{k-1}_{i=k} a_i = \prod_{i\in\emptyset} a_i = 1$.

\section{Definition of the Model with Cutoffs}

In this brief section we establish selfadjointness of  $H_\Lambda$.

\begin{Th}\label{ThSA}
Let $G^{(j)}_{\Lambda}$, $j=1,2$, be defined as in \eqref{KernelApprox} and suppose Hypothesis~\ref{Hypothesis-h},~\ref{Hypothesis-chi} and~\ref{MainHypothesis}.  Then the operator defined in \eqref{RegularisedHamiltonian} is self-adjoint with domain $\mathscr{D}( H_{\Lambda}) = \mathscr{D}( H_0)$, and essentially self-adjoint on $\mathscr H_\mathrm{fin}$. Furthermore
\[
H_\Lambda \geq  - C_\Lambda, \qquad \textup{where} \quad C_\Lambda =   1+
\int \Bigl(1+\frac1{\wb(q)}\Bigr) \Bigl(\bigl|G^{(1)}_{\Lambda}(k,q)\bigr|^2 +  \bigl|G^{(2)}_{\Lambda}(k,q)\bigr|^2\Bigr)dk dq.
\]
\end{Th}
\begin{proof}
Let us consider the following quadratic form associated to $H^{\cf \ab}(G^{(1)}_{\Lambda})$: 
\begin{align*}
 & \Bigl\langle \phi \Big| \int G^{(1)}_{\Lambda}(k,q)\cf(k)\ab(q) dk dq \psi  \Bigr\rangle    =  \int \overline{G^{(1)}_{\Lambda}(k,q)} \bigl\langle \af(k)\phi \big|  \ab(q)  \psi  \bigr\rangle dk dq  \\
 & \qquad =  \int  \Bigl\langle\int \overline{G^{(1)}_{\Lambda}(k,q)} \af(k) dk\phi \Big|  \ab(q)  \psi    \Bigr\rangle dq  
   =  \int  \Bigl\langle \af\Bigl(G^{(1)}_{\Lambda}(.,q) \Bigr)\phi \Big|  \ab(q)  \psi   \Bigr\rangle dq.  \\
  \end{align*}
  Applying Cauchy-Schwarz inequality leads to
  \begin{align*}
& \Bigl| \Bigl\langle \phi \Big| \int G^{(1)}_{\Lambda}(k,q)\cf(k)\ab(q) dk dq \psi  \Bigr\rangle \Bigr|      \leq  \int  \bigl\| \af\bigl( G_{\Lambda}(.,q) \bigr)\phi \bigr\| \bigl\|  \ab(q)  \psi  \bigr\|   dq  
  \leq  \int \bigl\| G^{(1)}_{\Lambda}(.,q)  \bigr\|  \bigl\|  \ab(q)  \psi  \bigr\|   dq \bigl\| \phi \bigr\| \\
&\qquad    \leq \int \Bigl\| \frac{G^{(1)}_{\Lambda}(.,q)}{\wb(q)^{\frac12}}  \Bigr\|  \left\|  \wb(q)^{\frac12} \ab(q)  \psi  \right\|   dq \bigl\| \phi \bigr\| 
       \leq  \biggl(\int  \frac{\bigl| G^{(1)}_{\Lambda}(k,q)\bigr|^2}{\wb(q)} dk dq \biggr)^{\frac12} \bigl\| H^{\frac12}_0  \psi  \bigr\|    \bigl\| \phi \bigr\|.
\end{align*}  

To treat the adjoint operator, we proceed in the following way:
\begin{align*}
 &\biggl\| \int \overline{G^{(1)}_{\Lambda}(k,q)}\af(k)\cb(q) dk dq \psi  \biggr\|^2 
  =  \biggl| \int \Bigl\langle  \af\bigl(G^{(1)}_{\Lambda}(.,q')\bigr)\cb(q')  \psi\Big| \af\bigl(G^{(1)}_{\Lambda}(.,q)\bigr)\cb(q) \psi\Bigr> dq'dq\biggr|\\
 & \qquad \leq \biggl| \int \Bigl<  \af\bigl(G^{(1)}_{\Lambda}(.,q')\bigr)\ab(q)  \psi\Big| \af\bigl(G^{(1)}_{\Lambda}(.,q)\bigr)\ab(q') \psi\Bigr\rangle dq'dq\biggr|
  +\biggl| \int \Bigl<  \af\bigl(G^{(1)}_{\Lambda}(.,q)\bigr) \psi\Big| \af\bigl(G^{(1)}_{\Lambda}(.,q)\bigr) \psi\Bigr> dq\biggr|\\
 & \qquad \leq  \biggl(\int\frac{\bigl|G^{(1)}_{\Lambda}(k,q)\bigr|^2}{\wb(q)}dq dk\biggr) \bigl\| H^{\frac{1}{2}}_0 \psi\bigr\|^2 +  \biggl(\int \bigl|G^{(1)}_{\Lambda}(k,q)\bigr|^2 dq dk\biggr) \bigl\|\psi \bigr\|^2.
  \end{align*} 
  Here we used Cauchy-Schwarz and the fact that smeared fermionic annihilation (and creation) operators are bounded: $\|\af(F)\| \leq \|F\|$, for $F\in L^2(\mathbb R^d)$.
Repeating this procedure for the remaining terms in $H_\mathrm{I}=H_\mathrm{I}(G^{(1)}_\Lambda,G^{(2)}_\Lambda)$ leads to:
\begin{align*}
\bigl\| H_\mathrm{I} \psi \bigr\|^2 & \leq  2\int \frac1{\wb(q)}\Bigl(\bigl|G^{(1)}_{\Lambda}(k,q)\bigr|^2 + \bigl|G^{(2)}_{\Lambda}(k,q)\bigr|^2\Bigr)dq dk\, \bigl\|H^{\frac12}_0 \psi\bigr\|^2 \\
& \quad  + \int \Bigl(\bigl|G^{(1)}_{\Lambda}(k,q)\bigr|^2 + \bigl|G^{(2)}_{\Lambda}(k,q)\bigr|^2\Bigr) dq dk\, \bigl\|\psi \bigr\|^2.\\
\end{align*}
The Kato-Rellich Theorem now implies the stated result. The lower bound on $H_\Lambda$ follows easily from the estimate above.
\end{proof}

In order to renormalise \eqref{RegularisedHamiltonian}, we introduce the following energy counterterm: 
\begin{equation}
\label{CounterTerm}
E^{(2)}_{\Lambda} = - \int\frac{\bigl|G^{(2)}_{\Lambda}(k,q)\bigr|^2}{\wf(k)+\wb(q)} dk dq.
\end{equation}
Physically, $E^{(2)}_{\Lambda}$ can be interpreted as the second order term in perturbation theory of the ground state energy.  The authors are currently working on a method which may be able to include higher order counterterms. In this connection, see \cite{Lampart2020,Lampart2022}, where the next order is taken into account for a different model. The resolvents of the following operators will be studied: 
\[
H_{\Lambda}-E^{(2)}_{\Lambda}.
\]
We already know from Theorem~\ref{ThSA} that $H_\Lambda-E^{(2)}_\Lambda$ is bounded below and consequently, there exists $C>0$ such that for any  $z \in \mathbb{C}$ such that $\Re(z)<-C$ the operator $H_{\Lambda}-E^{(2)}_{\Lambda}-z$ is invertible. Consequently,
\begin{equation*}
R_{\Lambda}(z)  =  \bigl(H_{\Lambda} - E^{(2)}_{\Lambda}-z\bigr)^{-1}
 =  \bigl(H_0+H_\mathrm{I}\bigl(G^{(1)}_{\Lambda},G^{(2)}_\Lambda\bigr) - E^{(2)}_{\Lambda}-z\bigr)^{-1}
 =  R_0(z) \bigl(1 + \bigl(H_\mathrm{I}\bigl(G^{(1)}_{\Lambda},G^{(2)}_\Lambda\bigr) - E^{(2)}_{\Lambda}\bigr)R_0(z)\bigr)^{-1}. 
\end{equation*}
We will investigate the convergence properties of the series:
\begin{equation}
\label{Sum}
R_{\Lambda}(z) =  R_0(z)\sum^{\infty}_{k=0} \Bigl\{ - \bigl(H_\mathrm{I}\bigl(G_{\Lambda}^{(1)},G^{(2)}_\Lambda\bigr) - E^{(2)}_{\Lambda}\bigr)R_0(z)\Bigr\}^k.
\end{equation}
From Lemma~\ref{IE1}, we have:
\[
\Bigl\| \Bigl( H^{\af \ab}\bigl(G^{(2)}_{\Lambda}\bigr) + H^{\cf \ab}\bigl(G^{(1)}_{\Lambda}\bigr)\Bigr) R_0(z)^{\frac12}\Bigr\| \leq \biggl(\int  \frac{\bigl| G^{(1)}_{\Lambda}(k,q)\bigr|^2}{\wb(q)}  dk dq \biggr)^{\frac12} +  \biggl(\int  \frac{\bigl| G^{(2)}_{\Lambda}(k,q)\bigr|^2}{\wb(q)}   dk dq \biggr)^{\frac12}
\]
and 
\[
\Bigl\|R_0(z)^{\frac12} \Bigl( H^{\af \cb}\bigl(G^{(1)}_{\Lambda}\bigr) + H^{\cf \cb}\bigl(G^{(2)}_{\Lambda}\bigr)\Bigr)\Bigr\| \leq  \biggl(\int  \frac{\bigl|G^{(1)}_{\Lambda}(k,q)\bigr|^2}{\wb(q)}   dk dq \biggr)^{\frac12} +  \biggl(\int  \frac{\bigl| G^{(2)}_{\Lambda}(k,q)\bigr|^2}{\wb(q)}  dk dq \biggr)^{\frac12}.
\]
Therefore:
\begin{align*}
\bigl\|R_{\Lambda}(z)\bigr\| &\leq  \bigl\|R_0(z)\bigr\| \sum^{\infty}_{k=0} 5^k C^{k}_{\Lambda} \bigl\| R_0(z) \bigr\|^{\frac{k}{2}},
\end{align*}
where $C_{\Lambda}$ has been defined in Proposition~\ref{ThSA}. Consequently, the series in \eqref{Sum} is convergent in norm, for $z \in \mathbb{C}$ with
\[
 \Re(z) < - 25 C_{\Lambda}^2.
\]


\section{Proof of Theorem \ref{MainTh}}

We advise the reader to read Appendix A to get familiar with the notations used in this section. We refer in particular to the definition of the sets $G_i$ with $i\in\llbracket 1,7\rrbracket$. We start this section proving the following theorem: 
\begin{Th}
Assume that the elements of the Kernels defined in \eqref{KernelApprox} fulfil Hypothesis~\ref{Hypothesis-h},~\ref{Hypothesis-chi}, and~\ref{MainHypothesis}. Assume moreover that $p>\frac{d}{2}-\frac34$. Then there exists a self-adjoint operator $H$, bounded from below, such that  $H_{\Lambda} - E^{(2)}_{\Lambda}$ converges in norm-resolvent sense to $H$. The renormalised Hamiltonian $H$ does not depend on $\chi$.
\end{Th}

\begin{proof}
According to Theorem \ref{ReorderingTh} the series \eqref{Sum} can be written as: 
\begin{equation}
\label{SumToStudy}
R_0(z)+ \sum^{\infty}_{k=1} {\underset{\text{If~} j_{l'} = 1, \text{~then~} j_{l'+1} \neq 2, \text{~and~}4}{ \underset{\text{If~} j_{l'} = 3, \text{~then~} j_{l'+1} \neq 2}{\sum_{\underset{j_1,\dots, j_l \in \llbracket 1,7 \rrbracket, ~~ \sum^l_{l'=1} \nop(j_{l'}) =k}{T^{(j_{l'})}\in G_{j_l'}}}} }} \hspace{-1cm} \hspace{0.5 cm} R_0(z) \hspace{0.5 cm} \prod^l_{k'=1} \Bigl[ T^{  (j_{k'})}_{\Lambda} R_0(z) \Bigr].
\end{equation}
We recall that the function $\nop$ is defined in Appendix A and that $\nop(j)$  counts the number of kernels involved in sequences of operators belonging to $G_j$. Theorem \ref{ReorderingTh} ensures, moreover, that the series \eqref{SumToStudy} is absolutely convergent for $\Re(z)< -25C_{\Lambda}^2$.  We start focusing on a specific summand:
\[
R_0(z) \ \prod^l_{k'=1} \Bigl[ T^{  (j_{k'})}_{\Lambda} R_0(z) \Bigr].
\]
Let
\begin{equation}\label{Connumu}
\begin{aligned}
 &\nu_1 = 0 \text{~and~} \mu_1 = \frac{3}{4}, 
 &\nu_2 = \frac34 \text{~and~} \mu_2 = 0,\\ 
& \nu_3 = 0 \text{~and~} \mu_3 = \frac{1}{2},
& \nu_4 = \frac12 \text{~and~} \mu_4 = 0,\\
& \nu_5 = \frac14 \text{~and~} \mu_5 = \frac 14,
& \nu_6 = 0 \text{~and~} \mu_6 = \frac{1}{4},\\
&\nu_7 = \frac14 \text{~and~} \mu_7 = 0.
\end{aligned}
\end{equation}
Write 
\begin{equation}
\label{ReorderedSumTh3}
R_0(z) \  \prod^l_{k'=1} \Bigl[ T^{  (j_{k'})}_{\Lambda} R_0(z) \Bigr] =  R_0(z)^{1-\nu_{j_1}} \prod^l_{k'=1} \Bigl[\Bigl(  R_0(z)^{\nu_{j_{k'}}}T^{  (j_{k'})}_{\Lambda} R_0(z)^{\mu_{j_{k'}}} \Bigr) R_0(z)^{1-\nu_{j_{k'+1}}-\mu_{j_{k'}}}\Bigr],
\end{equation}
where all the powers of free resolvents are non-negative, due to the constraints on the numbers $j_1,j_2,\dotsc,j_l$, which ensure that $\mu_{j_{k'}}+\nu_{j_{k'+1}} \leq 1$ for $k'=1,2,\dotsc,l$, with the convention that $\nu_{j_{l+1}}=0$. The exponents $\mu_j,\nu_j$, $j=1,2,\dotsc,7$, have been chosen such that the estimates in Appendix~\ref{InitialisationEstimates} apply.

   First,  the lemmas \ref{IE1}, \ref{IE2}, \ref{IE3}, \ref{IE4}, \ref{IE5}, \ref{IE6} and \ref{RegUniformKernel} can be used to prove that the summand \eqref{ReorderedSumTh3} is uniformly bounded with respect to $\Lambda$. Indeed, if one defines 
 \begin{equation}
 M_z = \sup_{\Lambda \geq 1 }\max\Bigl\{K^{(1)}_{z,\frac34}\bigl( G^{(\sharp)}_{\Lambda}\bigr), K^{(2)}_{z,\frac12}\bigl( G^{(\sharp)}_{\Lambda},G^{(\sharp)}_{\Lambda}\bigr),K^{(3)}_{z,\frac14}\bigl( G^{(\sharp)}_{\Lambda},G^{(\sharp)}_{\Lambda},G^{(\sharp)}_{\Lambda}\bigr)\Bigr\}
\end{equation}  
 then there exists $C>0$ such that
 \begin{equation}
 \forall j \in \llbracket 1, 7 \rrbracket: \qquad \bigl\| R_0(z)^{\nu_j} T^{(j)}_{\Lambda} R_0(z)^{\mu_j}\bigr\| \leq C M_z.
\end{equation}  
  Therefore, 
\begin{align*}
\Bigl\| R_0(z) \ \prod^l_{k'=1} \Bigl[ T^{  (j_{k'})}_{\Lambda} R_0(z) \Bigr] \Bigr\| & \leq  \Bigl\| R_0(z)^{1-\nu_{j_1}} \Bigr\| \prod^l_{k'=1} \Bigl\|  R_0(z)^{\nu_{j_{k'}}}T^{  (j_{k'})}_{\Lambda} R_0(z)^{\mu_{j_{k'}}} \Bigr\| \Bigl\| R_0(z)^{1-\nu_{j_{k'+1}}-\mu_{j_{k'}}}\Bigr\|\\
 & \leq   \Bigl\| R_0(z)^{1-\nu_{j_1}} \Bigr\| \prod^l_{k'=1} C M_z \Bigl\| R_0(z)^{1-\nu_{j_{k'+1}}-\mu_{j_{k'}}}\Bigr\|.
\end{align*} 
 
 Let us moreover note that 
 \begin{equation}
 \sum^{l}_{k=1}\left( \nu_{j_k} +\mu_{j_k} \right)\leq \frac34 l
 \end{equation}
and consequently:
\begin{align*}
\Bigl\| R_0(z) \  \prod^l_{k'=1} \Bigl[ T^{  (j_{k'})}_{\Lambda} R_0(z) \Bigr] \Bigr\|  
 & \leq   C^l M^l_z \bigl\| R_0(z) \bigr\|^{1+\frac{l}{4}}.
\end{align*}  
As a consequences, the series \eqref{SumToStudy}, and then the resolvent of the Hamiltonian, can be estimated in the following way:
\begin{align*}
\bigl\| R_{\Lambda}(z)  \bigr\|  
 & \leq  \|R_0(z)\|+\sum^{\infty}_{k=1} {\underset{\text{If~} j_{l'} = 1, \text{~then~} j_{l'+1} \neq 2, \text{~and~}4}{ \underset{\text{If~} j_{l'} = 3, \text{~then~} j_{l'+1} \neq 2}{\sum_{\underset{j_1,\dots, j_l \in \llbracket 1,7 \rrbracket, ~~ \sum^l_{l'=1} \nop(j_{l'}) =k}{T^{(j_{l'})}\in G_{j_l'}}}} }} \hspace{-1cm}  \Bigl\|  R_0(z) \  \prod^l_{k'=1} \Bigl[ T^{  (j_{k'})}_{\Lambda} R_0(z) \Bigr] \Bigr\|\\
 & \leq  \|R_0(z)\|+\sum^{\infty}_{k=1} {\underset{\text{If~} j_{l'} = 1, \text{~then~} j_{l'+1} \neq 2, \text{~and~}4}{ \underset{\text{If~} j_{l'} = 3, \text{~then~} j_{l'+1} \neq 2}{\sum_{\underset{j_1,\dots, j_l \in \llbracket 1,7 \rrbracket, ~~ \sum^l_{l'=1} \nop(j_{l'}) =k}{T^{(j_{l'})}\in G_{j_l'}}}} }} \hspace{-1cm}C^l M^l_z \bigl\| R_0(z) \bigr\|^{1+\frac{l}{4}}  
 \end{align*}
 Moreover, there exists $\lambda_0<-1$, such that for any $z$ fulfilling $\Re(z)< \lambda_0$, we have  $C M_z<1$ and $\Re(z)<-12^4$. Consequently,
 \begin{align*}
 \bigl\| R_{\Lambda}(z)  \bigr\|   & \leq  \|R_0(z)\|+ \sum^{\infty}_{k=1} \bigl\| R_0(z) \bigr\|^{1+\frac{k}{4}} {\underset{\text{If~} j_{l'} = 1, \text{~then~} j_{l'+1} \neq 2, \text{~and~}4}{ \underset{\text{If~} j_{l'} = 3, \text{~then~} j_{l'+1} \neq 2}{\sum_{\underset{j_1,\dots, j_l \in \llbracket 1,7 \rrbracket, ~~ \sum^l_{l'=1} \nop(j_{l'}) =k}{T^{j_{l'}}\in G_{j_l'}}}} }} \hspace{-1cm} 1    \\
 &\leq  \displaystyle   \sum^{\infty}_{k=0} 12^k  \bigl\| R_0(z) \bigr\|^{1+\frac{k}{4}} =  \frac{\| R_0(z) \|}{1-12  \| R_0(z) \|^{\frac14}}.
\end{align*} 
For later use, we remark that for $z$ with $\Re(z)<\lambda_0$, we may similarly derive the estimate
\begin{equation}\label{ReconThm3}
\bigl\|z\bigl( R_{\Lambda}(z)- R_0(z)\bigr)  \bigr\| \leq  |z|^{-\frac14}\bigl(1-12  \| R_0(z) \|^{\frac14}\bigr)^{-1}.
\end{equation}
We conclude that for any $\Lambda$, the sequences $R_{\Lambda}(z)$ can be expanded as an absolutely convergent Neumann series for $z$ with $\Re(z)<\lambda_0$.

We then claim that there exists $\lambda_0'\leq \lambda_0$, such that for $\Re(z)<\lambda_0'$, the limit $\displaystyle \lim_{\Lambda \to \infty}R_{\Lambda}(z)$ exists. Indeed, for two real numbers $\Lambda_1, \Lambda_2$ we have: 
\begin{align*}
R_{\Lambda_1}(z) - R_{\Lambda_2}(z) &= \displaystyle\sum^{\infty}_{k=1}  {\underset{\text{If~} j_{l'} = 1, \text{~then~} j_{l'+1} \neq 2, \text{~and~}4}{ \underset{\text{If~} j_{l'} = 4, \text{~then~} j_{l'+1} \neq 2}{\sum_{\underset{j_1,\dots, j_l \in \llbracket 1,7 \rrbracket, ~~ \sum^l_{l'=1} \nop(j_{l'}) =k}{T^{j_{l'}}\in G_{j_l'}}}} }} \hspace{-1cm}  R_0(z) \sum^{l}_{i=1}\prod^{i-1}_{k'=1}\Bigl(  T^{  (j_{k'})}_{\Lambda_2} R_0(z) \Bigr)  \Bigl(  \Bigl\{T^{  (j_{i})}_{\Lambda_1}-T^{  (j_{i})}_{\Lambda_2}\Bigr\} R_0(z) \Bigr)  \prod^{l}_{k'=i+1}\Bigl(  T^{  (j_{k'})}_{\Lambda_1} R_0(z) \Bigr) .
\end{align*}
We again start focusing on a specific summand, which is here: 
\begin{align}
\label{NewSumToEst}
 R_0(z) \prod^{i-1}_{k'=1}\Bigl(  T^{  (j_{k'})}_{\Lambda_2} R_0(z) \Bigr)  \Bigl(  \Bigl\{T^{  (j_{i})}_{\Lambda_1}-T^{  (j_{i})}_{\Lambda_2}\Bigr\} R_0(z) \Bigr)  \prod^{l}_{k'=i+1}\Bigl(  T^{  (j_{k'})}_{\Lambda_1} R_0(z) \Bigr).
\end{align}
As before, we split the resolvents using the exponents $\nu_j,\mu_j$, $j=1,2,\dotsc,7$, from \eqref{Connumu}.
Then \eqref{NewSumToEst} can be written in the following way:
\begin{align*}
&  R_0(z)^{1-\nu_{j_1}} \prod^{i-1}_{k'=1}\Bigl(  R_0(z)^{\nu_{j_{k'}}}T^{  (j_{k'})}_{\Lambda_2} R_0(z)^{\mu_{j_{k'}}} \Bigr) R_0(z)^{1-\nu_{j_{k'+1}}-\mu_{j_{k'}}}\\
 & \quad  \Bigl(  R_0(z)^{\nu_{j_{i}}}\Bigl\{T^{  (j_{i})}_{\Lambda_1}-T^{  (j_{i})}_{\Lambda_2}\Bigr\} R_0(z)^{\mu_{j_{i}}} \Bigr) R_0(z)^{1-\nu_{j_{i+1}}-\mu_{j_{i}}}\\
 & \quad \prod^{l}_{k'=i+1}\Bigl(  R_0(z)^{\nu_{j_{k'}}}T^{  (j_{k'})}_{\Lambda_1} R_0(z)^{\mu_{j_{k'}}} \Bigr) R_0(z)^{1-\nu_{j_{k'+1}}-\mu_{j_{k'}}} ,
\end{align*}
again with the convention that $\nu_{l+1}=0$.
We have to estimate the operators $T^{  (j_{i})}_{\Lambda_1}-T^{  (j_{i})}_{\Lambda_2}$. If $j_i$ is equal to $1$ or $2$ then $T^{  (j_{i})}_{\Lambda_{\sharp}}$ is of the form $H^{\ab^{\sharp }\af^{\sharp}}(G^{(\sharp)}_{\Lambda_{\sharp}})$. Consequently
\begin{align*}
\Bigl\|  R_0(z)^{\nu_{j_{i}}}\Bigl\{T^{  (j_{i})}_{\Lambda_1}-T^{  (j_{i})}_{\Lambda_2}\Bigr\} R_0(z)^{\mu_{j_{i}}}\Bigr\| & = \Bigl\| R_0(z)^{\nu_{j_{i}}}\Bigl\{H^{\ab^{\sharp }\af^{\sharp}}\bigl(G^{(\sharp)}_{\Lambda_{1}}\bigr)-H^{\ab^{\sharp }\af^{\sharp}}\bigl(G^{(\sharp)}_{\Lambda_{2}}\bigr)\Bigr\} R_0(z)^{\mu_{j_{i}}}\Bigr\|\\
& =  \Bigl\|R_0(z)^{\nu_{j_{i}}} H^{\ab^{\sharp }\af^{\sharp}}\bigl(G^{(\sharp)}_{\Lambda_{1}}-G^{(\sharp)}_{\Lambda_{2}}\bigr) R_0(z)^{\mu_{j_{i}}} \Bigr\|\\
& \leq  K^{(1)}_{z,\nu_{j_{i}}+ \mu_{j_{i}} }\bigl(G^{(\sharp)}_{\Lambda_{1}}-G^{(\sharp)}_{\Lambda_{2}}\bigr),
\end{align*}
where Lemma \ref{IE1} has been used. In the same way, if $j_i$ is equal to $3$ or $4$ then $T^{  (j_{i})}_{\Lambda_{1}} - T^{  (j_{i})}_{\Lambda_{2}}$ can be equal to 
\begin{equation*}
    H^{\ab \af }\bigl(G^{(2)}_{\Lambda_{1}}\bigr) R_0(z) H^{\cb \af }\bigl(G^{(1)}_{\Lambda_{1}}\bigr) - H^{\ab \af }\bigl(G^{(2)}_{\Lambda_{2}}\bigr) R_0(z) H^{\cb \af }\bigl(G^{(1)}_{\Lambda_{2}}\bigr)
\end{equation*}
or 
\begin{equation*}
    H^{\ab \cf }\bigl(G^{(1)}_{\Lambda_{1}}\bigr) R_0(z) H^{\cb \cf }\bigl(G^{(2)}_{\Lambda_{1}}\bigr) - H^{\ab \cf }\bigl(G^{(1)}_{\Lambda_{2}}\bigr) R_0(z) H^{\cb \cf }\bigl(G^{(2)}_{\Lambda_{2}}\bigr).
\end{equation*}

Let us treat the first case
\begin{align*}
&\Bigl\|  R_0(z)^{\nu_{j_{i}}}\Bigl\{T^{  (j_{i})}_{\Lambda_1}-T^{  (j_{i})}_{\Lambda_2}\Bigr\} R_0(z)^{\mu_{j_{i}}}\Bigr\|\\
 & =  \Bigl\| R_0(z)^{\nu_{j_{i}}}\Bigl\{H^{\ab \af}\bigl(G^{(2)}_{\Lambda_{1}}\bigr) R_0(z) H^{\cb \af}\bigl(G^{(1)}_{\Lambda_{1}}\bigr)-  H^{\ab\af}\bigl(G^{(2)}_{\Lambda_{2}}\bigr) R_0(z) H^{\cb\af}\bigl(G^{(1)}_{\Lambda_{2}}\bigr)\Bigr\} R_0(z)^{\mu_{j_{i}}}\Bigr\|\\
 & =  \Bigl\| R_0(z)^{\nu_{j_{i}}}\Bigl\{H^{\ab\af}\bigl(G^{(2)}_{\Lambda_{1}}-G^{(2)}_{\Lambda_{2}}\bigr) R_0(z) H^{\cb\af}\bigl(G^{(1)}_{\Lambda_{1}}) + H^{\ab\af}(G^{(2)}_{\Lambda_{2}}\bigr) R_0(z) H^{\cb\af}\bigl(G^{(1)}_{\Lambda_{1}}-G^{(1)}_{\Lambda_{2}}\bigr)\Bigr\} R_0(z)^{\mu_{j_{i}}}\Bigr\|,
 \end{align*}
 which can be estimated using Lemma \ref{IE3} by
 \begin{equation*}
   K^{(2)}_{z,\nu_{j_{i}}+ \mu_{j_{i}} }\bigl(G^{(2)}_{\Lambda_{1}}-G^{(2)}_{\Lambda_{2}}, G^{(1)}_{\Lambda_{1}} \bigr) + K^{(2)}_{z,\nu_{j_{i}}+ \mu_{j_{i}} }\bigl(G^{(2)}_{\Lambda_{2}}, G^{(1)}_{\Lambda_{1}}-G^{(1)}_{\Lambda_{2}} \bigr).
\end{equation*}
If $j_i$ is equal to $5$, then either 
\begin{equation}
T^{(j_i)}_{\Lambda}  =  H^{\ab \cf}\bigl(G^{(1)}_{\Lambda}\bigr) R_0(z) H^{\cb \af}\bigl(G^{(1)}_{\Lambda}\bigr),
\end{equation}
which, by using Lemma \ref{IE4}, can be treated in the same way as before,  or
\begin{equation}
T^{(j_i)}_{\Lambda} = H^{\ab \af }\bigl(G^{(2)}_{\Lambda}\bigr) R_0(z) H^{ \cb \cf }\bigl(G^{(2)}_{\Lambda}\bigr) + E\bigl(G^{(2)}_{\Lambda},G^{(2)}_{\Lambda}\bigr).
\end{equation}
Recall from \eqref{EFG} the notation $E(F,G)$ and note that $E_\Lambda^{(2)}=E(G_\Lambda^{(2)},G_\Lambda^{(2)})$.
In this case: 
\begin{align*}
&\Bigl\|  R_0(z)^{\nu_{j_{i}}}\Bigl\{T^{  (j_{i})}_{\Lambda_1}-T^{  (j_{i})}_{\Lambda_2}\Bigr\} R_0(z)^{\mu_{j_{i}}}\Bigr\|\\
 & \quad =  \Bigl\| R_0(z)^{\nu_{j_{i}}}\Bigl\{H^{\ab \af }\bigl(G^{(2)}_{\Lambda_1}\bigr) R_0(z) H^{ \cb \cf }\bigl(G^{(2)}_{\Lambda_1}\bigr) + E\bigl(G^{(2)}_{\Lambda_1},G^{(2)}_{\Lambda_1}\bigr) \\
 & \qquad -  H^{\ab \af }\bigl(G^{(2)}_{\Lambda_2}\bigr) R_0(z) H^{ \cb \cf }\bigl(G^{(2)}_{\Lambda_2}\bigr) - E\bigl(G^{(2)}_{\Lambda_2},G^{(2)}_{\Lambda_2}\bigr)\Bigr\} R_0(z)^{\mu_{j_{i}}}\Bigr\|\\
 & \quad  =  \Bigl\| R_0(z)^{\nu_{j_{i}}}\Bigl\{H^{\ab \af }\bigl(G^{(2)}_{\Lambda_1}-G^{(2)}_{\Lambda_2}\bigr) R_0(z) H^{ \cb \cf }\bigl(G^{(2)}_{\Lambda_1}\bigr) + E\bigl(G^{(2)}_{\Lambda_1}-G^{(2)}_{\Lambda_2},G^{(2)}_{\Lambda_1}\bigr) \\
 & \qquad +  H^{\ab \af }\bigl(G^{(2)}_{\Lambda_2}\bigr) R_0(z) H^{ \cb \cf }\bigl(G^{(2)}_{\Lambda_1}-G^{(2)}_{\Lambda_2}\bigr) + E\bigl(G^{(2)}_{\Lambda_2},G^{(2)}_{\Lambda_1}-G^{(2)}_{\Lambda_2}\bigr)\Bigr\} R_0(z)^{\mu_{j_{i}}}\Bigr\|,
 \end{align*}
 which by Lemma \ref{IE2} can be estimates to be less than or equal to
\begin{equation*}
  K^{(2)}_{z,\nu_{j_{i}}+ \mu_{j_{i}} }\bigl(G^{(2)}_{\Lambda_{1}}-G^{(2)}_{\Lambda_{2}}, G^{(2)}_{\Lambda_{1}} \bigr) + K^{(2)}_{z,\nu_{j_{i}}+ \mu_{j_{i}} }\bigl(G^{(2)}_{\Lambda_{2}}, G^{(2)}_{\Lambda_{1}}-G^{(2)}_{\Lambda_{2}} \bigr).
\end{equation*}
Finally, if $j_i$ is equal to $6$ or $7$ then $T^{  (j_{i})}_{\Lambda_{\sharp}}$ can be one of the operators considered in Lemma \ref{IE5} and \ref{IE6}.  We will treat the following case: 
\[
 H^{\ab \af}\bigl(G^{(2)}_{\Lambda_{1}}\bigr) R_0(z)  H^{\ab \cf}\bigl(G^{(1)}_{\Lambda_{1}}\bigr) R_0(z) H^{\cb \cf} \bigl(G^{(2)}_{\Lambda_{1}}\bigr) - H^{\ab \af}\bigl(G^{(2)}_{\Lambda_{2}}\bigr) R_0(z)  H^{\ab \cf}\bigl(G^{(1)}_{\Lambda_{2}}\bigr) R_0(z) H^{\cb \cf} \bigl(G^{(2)}_{\Lambda_{2}}\bigr),
\]
the other ones being similar. Then, by using Lemma \ref{IE5}:
\begin{align*}
&\Bigl\|  R_0(z)^{\nu_{j_{i}}}\Bigl\{T^{  (j_{i})}_{\Lambda_1}-T^{  (j_{i})}_{\Lambda_2}\Bigr\} R_0(z)^{\mu_{j_{i}}}\Bigr\|\\
 & \quad  =  \Bigl\| R_0(z)^{\nu_{j_{i}}}\Bigl\{H^{\ab\af}\bigl(G^{(2)}_{\Lambda_{1}}\bigr) R_0(z) H^{\ab\cf}\bigl(G^{(1)}_{\Lambda_{1}}\bigr)R_0(z)H^{\cb\cf}\bigl(G^{(2)}_{\Lambda_{1}}\bigr)\\
 & \qquad - H^{\ab\af}\bigl(G^{(2)}_{\Lambda_{2}}\bigr) R_0(z) H^{\ab\cf}\bigl(G^{(1)}_{\Lambda_{2}}\bigr) R_0(z) H^{\cb\cf}\bigl(G^{(2)}_{\Lambda_{2}}\bigr)\Bigr\} R_0(z)^{\mu_{j_{i}}}\Bigr\|\\
  & \quad = \Bigl\| R_0(z)^{\nu_{j_{i}}}\Bigl\{H^{\ab\af}\bigl(G^{(2)}_{\Lambda_{1}}-G^{(2)}_{\Lambda_{2}}\bigr) R_0(z) H^{\ab\cf}\bigl(G^{(1)}_{\Lambda_{1}}\bigr)R_0(z)H^{\cb\cf}\bigl(G^{(2)}_{\Lambda_{1}}\bigr)\\
  & \qquad + H^{\ab\af}\bigl(G^{(2)}_{\Lambda_{2}}\bigr) R_0(z) H^{\ab\cf}\bigl(G^{(1)}_{\Lambda_{1}}-G^{(1)}_{\Lambda_{2}}\bigr)R_0(z)H^{\cb\cf}\bigl(G^{(2)}_{\Lambda_{1}}\bigr)\\
 & \qquad +  H^{\ab\af}\bigl(G^{(2)}_{\Lambda_{2}}\bigr) R_0(z) H^{\ab\cf}\bigl(G^{(1)}_{\Lambda_{2}}\bigr) R_0(z) H^{\cb\cf}\bigl(G^{(2)}_{\Lambda_{1}}-G^{(2)}_{\Lambda_{2}}\bigr)\Bigr\} R_0(z)^{\mu_{j_{i}}}\Bigr\|\\
& \quad \leq  \Bigl(64 C_s^{4+2(\nu_{j_{i}}+ \mu_{j_{i}})}+ 3C_s^{\frac{4+2(\nu_{j_{i}}+ \mu_{j_{i}})}{3}} + 2 C_s\Bigr) \\
&\qquad \Bigl(K^{(3)}_{z,\nu_{j_{i}}+ \mu_{j_{i}} }\bigl(G^{(2)}_{\Lambda_{1}}-G^{(2)}_{\Lambda_{2}}, G^{(1)}_{\Lambda_{1}},G^{(2)}_{\Lambda_{1}} \bigr) + K^{(3)}_{z,\nu_{j_{i}}+ \mu_{j_{i}} }\bigl(G^{(2)}_{\Lambda_{2}}, G^{(2)}_{\Lambda_{1}}-G^{(1)}_{\Lambda_{2}},G^{(2)}_{\Lambda_{1}} \bigr)
+  K^{(3)}_{z,\nu_{j_{i}}+ \mu_{j_{i}} }\bigl(G^{(2)}_{\Lambda_{2}}, G^{(1)}_{\Lambda_{2}}, G^{(2)}_{\Lambda_{1}}-G^{(2)}_{\Lambda_{2}} \bigr)\Bigr).
\end{align*}

 Consequently,  there exists $C>0$ such that
 \begin{equation*}
  \Bigl\|  R_0(z)^{\nu_{j_{i}}}\Bigl\{T^{  (j_{i})}_{\Lambda_1}-T^{  (j_{i})}_{\Lambda_2}\Bigr\} R_0(z)^{\mu_{j_{i}}}\Bigr\| \leq C 
  \begin{cases}  
      K^{(1)}_{z,\frac34}\bigl(G^{(\sharp)}_{\Lambda_{1}}-G^{(\sharp)}_{\Lambda_{2}}\bigr), & \textup{if } j_i \in \{ 1,2\}  
      \vspace{3mm}
      \\
     \begin{aligned}
        & K^{(2)}_{z,\frac12}\bigl(G^{(\sharp)}_{\Lambda_{1}}-G^{(\sharp)}_{\Lambda_{2}}, G^{(\sharp)}_{\Lambda_{1}} \bigr)   \\
         & \qquad + K^{(2)}_{z,\frac12}\bigl( G^{(\sharp)}_{\Lambda_{2}},G^{(\sharp)}_{\Lambda_{1}}-G^{(\sharp)}_{\Lambda_{2}} \bigr),
     \end{aligned}    & \textup{if } j_i \in \{ 3,4,5 \}
       \vspace{3mm}
     \\
     \begin{aligned}
         & K^{(3)}_{z,\frac14}\bigl(G^{(\sharp)}_{\Lambda_{1}}-G^{(\sharp)}_{\Lambda_{2}}, G^{(\sharp)}_{\Lambda_{1}}, G^{(\sharp)}_{\Lambda_{1}}\bigr)   \\
         & \qquad + K^{(3)}_{z,\frac14}\bigl( G^{(\sharp)}_{\Lambda_{2}},G^{(\sharp)}_{\Lambda_{1}}-G^{(\sharp)}_{\Lambda_{2}}, G^{(\sharp)}_{\Lambda_{1}}\bigr)\\
         & \qquad + K^{(3)}_{z,\frac14}\bigl( G^{(\sharp)}_{\Lambda_{2}}, G^{(\sharp)}_{\Lambda_{2}}, G^{(\sharp)}_{\Lambda_{1}}-G^{(\sharp)}_{\Lambda_{2}}\bigr),
     \end{aligned} & \textup{if } j_i \in \{ 6,7 \}.
  \end{cases}   
 \end{equation*}
 Let $\epsilon>0$, according to Lemma \ref{RegUniformKernel}, for any $a,b,c \in \{1,2\}$ there exists $R>0$ such that for any $\Lambda_1,\Lambda_2 \geq R$, we have:
 \[
\begin{aligned}
K^{(1)}_{z,\frac34}\bigl(G^{(a)}_{\Lambda_1}-G^{(a)}_{\Lambda_2} \bigr) & \leq  \epsilon, \qquad  &
K^{(2)}_{z,\frac12}\bigl(G^{(a)}_{\Lambda_1}-G^{(a)}_{\Lambda_2},G^{(b)}_{\Lambda_{1}}  \bigr) & \leq  \epsilon,   \\
K^{(2)}_{z,\frac12}\bigl(G^{(a)}_{\Lambda_{2}} ,G^{(b)}_{\Lambda_1}-G^{(b)}_{\Lambda_2} \bigr) & \leq \epsilon,  \qquad &
K^{(3)}_{z,\frac14}\bigl(G^{(a)}_{\Lambda_1}-G^{(a)}_{\Lambda_2},G^{(b)}_{\Lambda_{1}}, G^{(c)}_{\Lambda_{1}}  \bigr) & \leq  \epsilon, \\
K^{(3)}_{z,\frac14}\bigl(G^{(a)}_{\Lambda_{2}} ,G^{(b)}_{\Lambda_1}-G^{(b)}_{\Lambda_2},G^{(c)}_{\Lambda_{1}} \bigr) & \leq  \epsilon,   \qquad &
K^{(3)}_{z,\frac14}\bigl(G^{(a)}_{\Lambda_{2}} , G^{(b)}_{\Lambda_{2}} , G^{(c)}_{\Lambda_1}-G^{(c)}_{\Lambda_2} \bigr) & \leq  \epsilon.
\end{aligned}
\]
Consequently, there exists a constant $C$, such that
\begin{align*}
\biggl\| R_0(z) \prod^{i-1}_{k'=1}\Bigl(  T^{  (j_{k'})}_{\Lambda_2} R_0(z) \Bigr)  \Bigl(  \Bigl\{T^{  (j_{i})}_{\Lambda_1}-T^{  (j_{i})}_{\Lambda_2}\Bigr\} R_0(z) \Bigr)  \prod^{l}_{k'=i+1}\Bigl(  T^{  (j_{k'})}_{\Lambda_1} R_0(z) \Bigr)\biggr\| & \leq \epsilon\, C^l \prod^{i-1}_{k'=1}M_z  \prod^{l}_{k'=i+1}M_z \bigl\| R_0(z) \bigr\|^{1+\frac{l}{4}}.
\end{align*}

Moreover, there exists $\lambda_0'\leq \lambda_0$ such that for any $z\in\mathbb C$ fulfilling $\Re(z)<\lambda_0'$, we have $C M_z < 1$ and $\Re(z)<-12^4$. Consequently, for such $z$ we have
\begin{align*}
\bigl\|R_{\Lambda_1}(z) - R_{\Lambda_2}(z) \bigr\| & \leq   \epsilon \sum^{\infty}_{k=1} {\underset{\text{If~} j_{l'} = 1, \text{~then~} j_{l'+1} \neq 2, \text{~and~}4}{ \underset{\text{If~} j_{l'} = 4, \text{~then~} j_{l'+1} \neq 2}{\sum_{\underset{j_1,\dots, j_l \in \llbracket 1,7 \rrbracket, ~~ \sum^l_{l'=1} \nop(j_{l'}) =k}{T^{j_{l'}}\in G_{j_l'}}}} }}  \sum^{l}_{i=1}C^l\prod^{i-1}_{k'=1}M_z  \prod^{l}_{k'=i+1}M_z \bigl\| R_0(z) \bigr\|^{1+\frac{l}{4}} \\
& \leq   \epsilon\, C  \sum^{\infty}_{k=1} 12^k  k \bigl\| R_0(z) \bigr\|^{1+\frac{k}{4}}  =   \epsilon\, C   \frac{12\bigl\|R_0(z)\bigr\|^{\frac54}}{\bigl[1-12 \bigl\| R_0(z) \bigr\|^{\frac{1}{4}}\bigr]^2},
\end{align*}
which proves that $R_{\Lambda}(z)$ is a Cauchy sequence and hence $R_\Lambda(z)$ converges to a limiting bounded operator $R(z)$, as $\Lambda\to \infty$.

Recalling \eqref{ReconThm3}, we now conclude the construction of the renormalised Hamiltonian $H$ by appealing to Theorem~\ref{ThAW2017}.  To see that $H$ does not depend on the choice of $\chi$, let $\chi,\chi'$ be two cutoffs satisfying Hypthesis~\ref{Hypothesis-chi}. Denote by $H_\chi$ and $H_{\chi'}$ the two renormalised Hamiltonians. For $z\in\mathbb C$ with $\Re(z)<0$ small enough, compute
\[
\bigl(H_\chi-z\bigr)^{-1} - \bigl(H_{\chi'}-z\bigr)^{-1} = \lim_{\Lambda\to \infty}\Bigl\{ \bigl(H_{\chi,\Lambda}+E^{(2)}_{\chi,\Lambda}-z\bigr)^{-1} - \bigl(H_{\chi',\Lambda}+E^{(2)}_{\chi',\Lambda}-z\bigr)^{-1}\Bigr\}
\]
by inserting on the right hand side the resummed powerseries representations \eqref{SumToStudy} of the two UV cutoff resolvents. Performing a telescoping expansion of the terms as when we established that $R_\Lambda(z)=(H_\Lambda+E^{(2)}_\Lambda-z)^{-1}$ has the Cauchy property, one may conclude the result. 
\end{proof}

\begin{proof}[Proof of Theorem \ref{MainTh}]
What have been proved so far, is that for $p>\frac{d}{2}-\frac{3}{4}$, there exists $H$ which is the norm resolvent limit of $H_{\Lambda}-E^{(2)}_{\Lambda}$. It remains to prove that the counterterm $E^{(2)}_{\Lambda}$ can be replaced by $E_{\Lambda} = \inf\sigma(H_{\Lambda})$. 

Since $H_{\Lambda}-E^{(2)}_{\Lambda}$ converges in norm resolvent sense to $H$, we conclude from \cite[Theorem~VIII.20]{RS} that
\[
E_{\Lambda}- E_{\Lambda}^{(2)} \to \inf\sigma(H), \quad \textup{for } \Lambda\to \infty,
\]
which implies that $H_{\Lambda}-E_{\Lambda}$ converges in norm resolvent sense to $H - \inf\sigma(H)$ and completes the proof.
\end{proof}

\appendix

\section{Reordering Theorem}
\label{ReorderingTheorem}
The purpose of this section is to prove that the order of summation of \eqref{Sum} can be changed, which enable us to highlight the cancellations that occur. We start by defining the  following sets of linear operators:

\begin{align*}
G_{1} & =  \Bigl\{- H^{\ab \af}\bigl(G^{(2)}_{\Lambda}\bigr), - H^{\ab \cf}\bigl(G^{(1)}_{\Lambda}\bigr)  \Bigr\}\\
G_2 & =  \Bigl\{ -H^{\cb \cf}\bigl(G^{(2)}_{\Lambda}\bigr), -H^{\cb \af}\bigl(G^{(1)}_{\Lambda}\bigr) \Bigr\}\\
G_3 & =  \Bigl\{ H^{\ab \af}\bigl(G^{(2)}_{\Lambda}\bigr) R_0(z)H^{\cb \af}\bigl(G^{(1)}_{\Lambda}\bigr) \Bigr\}\\
G_4 & =  \Bigl\{  H^{\ab \cf}\bigl(G^{(1)}_{\Lambda}\bigr) R_0(z) H^{\cb \cf }\bigl(G^{(2)}_{\Lambda}\bigr) \Bigr\}\\
G_5 & =  \Bigl\{  \left(H^{\ab \af }\bigl(G^{(2)}_{\Lambda}\bigr) R_0(z) H^{ \cb \cf }\bigl(G^{(2)}_{\Lambda}\bigr) + E^{(2)}_{\Lambda}\right), H^{\ab \cf}\bigl(G^{(1)}_{\Lambda}\bigr) R_0(z) H^{\cb \af}\bigl(G^{(1)}_{\Lambda}\bigr)  \Bigr\}\\
G_6 & =  \Bigl\{-\left(H^{\ab \af}\bigl(G^{(2)}_{\Lambda}\bigr) R_0(z)H^{\ab \cf}\bigl(G^{(1)}_{\Lambda}\bigr) R_0(z) H^{\cb \cf}\bigl(G^{(2)}_{\Lambda}\bigr)\right)\\
& \qquad\ - \left( H^{\ab \cf}\bigl(G^{(2)}_{\Lambda}\bigr) R_0(z)H^{\ab \cf}\bigl(G^{(1)}_{\Lambda}\bigr) R_0(z) H^{\cb \cf}\bigl(G^{(1)}_{\Lambda}\bigr)  \right)  \Bigr\}\\
G_7 & =  \Bigl\{-\left(H^{\ab \af}\bigl(G^{(2)}_{\Lambda}\bigr) R_0(z)H^{\cb \af}\bigl(G^{(1)}_{\Lambda}\bigr) R_0(z) H^{\cb \cf}\bigl(G^{(2)}_{\Lambda}\bigr)\right),\\
&\qquad  - \left(H^{\ab \af}\bigl(G^{(2)}_{\Lambda}\bigr) R_0(z) H^{\cb \af}\bigl(G^{(1)}_{\Lambda}\bigr) R_0(z) H^{\cb \af}\bigl(G^{(1)}_{\Lambda}\bigr) \right)\Bigr\}
\end{align*}
Moreover, if an operator $T$ belongs to some $G_i$ with $i\in \llbracket 1,7\rrbracket$ we will denote it $T^{(i)}$. Finally, we introduce the function $\nop$  defined on $\llbracket 1,7 \rrbracket$ in the following way: 
\begin{align*}
\nop(1)=\nop(2)  &= 1,\\  
\nop(3)=\nop(4)=\nop(5) & =  2, \\
\qquad \nop(6)= \nop(7) &=  3.
\end{align*}
The idea is that the function $n_{\mathrm{Op}}$ counts the number of kernels involved in the sequences of operator in the sets $G_i$.

Let us give some informal motivations to the $G$ sets. The main idea of the proof of Theorem \ref{MainTh} is to control uniformly with respect to the ultraviolet cut off the sequences of operators involved in the Neumann series. At each order it is convenient to split these sequences into elements of finite length which can be controlled more easily. This is the purpose of  Theorem~\ref{ReorderingTh}. The question that one has now to answer is what sets of finite sequences of operators have to be chosen. 

From Lemma \ref{IE1} and \ref{RegUniformKernel} it can be seen that $H^{\ab\af}(G^{(2)}_{\Lambda}) R_0(z)^{\beta}$ and $R_0(z)^{\beta} H^{\cb \cf}(G^{(2)}_{\Lambda}) $ are uniformly bounded in norm, with respect to the ultraviolet cut off  $\Lambda$, if 
\begin{equation}\label{BasicConstraint}
2\beta + 2p > d.
\end{equation}
It implies that it may not be possible to control uniformly the sequence $H^{\ab\af}(G^{(2)}_{\Lambda})  R_0(z)H^{\cb \cf}(G^{(2)}_{\Lambda}) $ by using Lemma \ref{IE1}. It actually seems impossible as soon as $p\leq \frac{d}{2} - \frac{1}{2}$. Lemma \ref{IE2} enable us to overcome this difficulty. However, we see that a counterterm has to be included. What is the best regularity condition one can hope to obtain without introducing more counterterms? Let us consider, for $\alpha,\beta\in [0,1]$,  the following expression: 
\[
\Bigl\{H^{\ab\af}(G^{(2)}_{\Lambda})  R_0(z)^{1-\alpha}\Bigr\}\Bigl\{R_0(z)^{\alpha} \left(H^{\ab\af}(G^{(2)}_{\Lambda})  R_0(z)H^{\cb \cf}(G^{(2)}_{\Lambda})  + E^{(2)}_{\Lambda} \right)R_0(z)^{\beta} \Bigr\} \Bigl\{ R_0(z)^{1-\beta}H^{\cb \cf}(G^{(2)}_{\Lambda})\Bigr\}. 
\]
Recall the following facts from Lemma~\ref{IE1} and \ref{IE2} (recall the constraint \eqref{BasicConstraint}), as well as Lemma~\ref{RegUniformKernel}
\begin{itemize}
    \item $H^{\ab\af}(G^{(2)}_{\Lambda})  R_0(z)^{1-\alpha}$ can be estimated uniformly in $\Lambda$ if $2(1-\alpha) + 2p > d$,
    \item  $R_0(z)^{1-\beta}H^{\cb \cf}(G^{(2)}_{\Lambda})$ can be estimated uniformly in $\Lambda$ if $2(1-\beta) + 2p > d$,
    \item $R_0(z)^{\alpha} \left(H^{\ab\af}(G^{(2)}_{\Lambda})  R_0(z)H^{\cb \cf}(G^{(2)}_{\Lambda})  + E^{(2)}_{\Lambda} \right)R_0(z)^{\beta}$ can be estimated uniformly in $\Lambda$ if $4p + 2 +2(\alpha+\beta) > 2d$.
\end{itemize}
One may choose $\alpha,\beta\in [0,1]$, such that the above constraints are satisfied, provided $4p+3  > 2d$,
leading to the requirement
\[
p>\frac{d}{2} - \frac34.
\]
In fact, $\alpha=\beta=\frac14$ will work under the above constraint.

 However, other incompatible concatenations of sequences may appear, preventing us to reach this condition. One then has to check all the possible concatenations, which leads to the construction of the $G$ sets. From now on, let us assume that $p>\frac{d}{2} - \frac34$.  The sets $G_1$ and $G_2$  are natural to consider. Moreover, from  Lemma \ref{IE1} and \ref{RegUniformKernel} we know that 
\begin{itemize}
    \item For any $T_1\in G_1$, $T_1 R_0(z)^{\frac34}$ is uniformly bounded with respect to the cut off. 
    \item For any $T_2\in G_2$, $R_0(z)^{\frac34} T_2 $ is uniformly bounded with respect to the cut off.
\end{itemize}
$T_1$ and $T_2$ are then incompatible between them as these estimates cannot be used to control $T_1R_0(z)T_2$. By exploring all the possible combinations, we isolate the elements of $G_3$, $G_4$ and $G_5$. 
\begin{itemize}
\item According to Lemma \ref{IE3} and Lemma \ref{RegUniformKernel}, for any $T_3\in G_3$,  $T_3 R_0(z)^{\frac12}$ is uniformly bounded with respect to the cut off.
\item According to Lemma \ref{IE3} and Lemma \ref{RegUniformKernel}, for any $T_4\in G_4$,  $ R_0(z)^{\frac12} T_4$ is uniformly bounded with respect to the cut off.
\item According to Lemma \ref{IE2}, \ref{IE4} and \ref{RegUniformKernel}, for any $T_5\in G_5$,  $ R_0(z)^{\gamma} T_4R_0(z)^{\delta}$ with $\delta + \gamma = \frac12$ is uniformly bounded with respect to the cut off.
\end{itemize}
Consequently, the following two concatenations cannot be uniformly controlled with the above cited Lemma's:
\[
   T_1 R_0(z) T_4\qquad \textup{and} \qquad T_3 R_o(z) T_2,
\]
i.e., the elements of $G_6$ and $G_7$. From Lemma~\ref{IE5}, \ref{IE6} and~\ref{RegUniformKernel}, we can see that
\begin{itemize}
    \item For any $T_6\in G_6$, $T_6 R_0(z)^{\frac14}$ is uniformly bounded with respect to the cut off. 
    \item For any $T_7\in G_7$, $R_0(z)^{\frac14} T_7 $ is uniformly bounded with respect to the cut off.
\end{itemize}
These last sequences can be concatenated with any other sequences.

\begin{Thannex}
\label{ReorderingTh}
Assume that $\Re(z)<-25 C^2_{\Lambda}$,  then the series \eqref{Sum} is convergent and can be reordered in the following way:
\begin{equation}
\label{ReorderedSUM}
 R_0(z)+ \sum^{\infty}_{k=1} {\underset{\text{If~} j_{l'} = 1, \text{~then~} j_{l'+1} \neq 2, \text{~and~}4}{ \underset{\text{If~} j_{l'} = 3, \text{~then~} j_{l'+1} \neq 2}{\sum_{\underset{j_1,\dots, j_l \in \llbracket 1,7 \rrbracket, ~~ \sum^l_{l'=1} \nop(j_{l'}) =k}{T^{(j_{l'})}\in G_{j_l'}}}} }} \hspace{-1cm} \hspace{0.5 cm} R_0(z) \hspace{0.5 cm} \prod^l_{k'=1} \left[ T^{  (j_{k'})}_{\Lambda} R_0(z) \right],
\end{equation}
\end{Thannex}

\begin{proof}
The main line of the proof is similar to the analogous results proved in \cite{AW2017}. Let us introduce the following notations:
\[
\begin{aligned}
\widetilde{T}^{(1)} & = - H^{\ab \af}\bigl(G^{(2)}_{\Lambda}\bigr), & \qquad
\widetilde{T}^{(2)} & = - H^{\cb \cf}\bigl(G^{(2)}_{\Lambda}\bigr),\\
\widetilde{T}^{(3)} & =  - H^{\ab \cf}\bigl(G^{(1)}_{\Lambda}\bigr), & \qquad
\widetilde{T}^{(4)} & =  -H^{\cb \af}\bigl(G^{(1)}_{\Lambda}\bigr),\\
\widetilde{T}^{(5)} & =  E^{(2)}_{\Lambda}, & \qquad 
\widetilde{T}^{(6)} & =  H^{\ab \af}\bigl(G^{(2)}_{\Lambda}\bigr) R_0(z) H^{\cb \cf}\bigl(G^{(2)}_{\Lambda}\bigr),\\
\widetilde{T}^{(7)} & =  H^{\ab \af}\bigl(G^{(2)}_{\Lambda}\bigr) R_0(z) H^{\cb \cf}\bigl(G^{(2)}_{\Lambda}\bigr) + E^{(2)}_{\Lambda}. & &
\end{aligned}
\]
The series \eqref{Sum} can be written in the following way: 
\begin{equation}
\label{Step1}
R_0(z)+\sum^{\infty}_{k=1} \sum_{(j_1, \dots, j_k) \in \llbracket 1,\dots,5\rrbracket^{k}} R_0(z) \prod^k_{i=1}\Bigl[ \widetilde{T}^{(j_i)}R_0(z)\Bigr].
\end{equation}
We may introduce an analogous function of $\nop$ associated to these new operators:

\begin{align*}
\tnop(\llbracket 1, 4 \rrbracket) & =  1 \\
\tnop(\llbracket 5, 7 \rrbracket) & =  2.
\end{align*}
It is sufficient to show that the series \eqref{Step1} can be rewritten as: 
\begin{equation}
\label{Goal}
R_0(z)+\sum^{\infty}_{k=1} \sum_{\underset{\sum^l_{j'=1} \tnop(j_{l'}) = k}{\underset{\text{If~} j_{l'} = 1 \text{~then~} j_{l'+1}\neq 2}{(j_1, \dots, j_l) \in (\llbracket 1,\dots,4\rrbracket\cup\{7\})^{l}}}} R_0(z) \prod^l_{i=1}\Bigl[ \widetilde{T}^{(j_i)}R_0(z)\Bigr].
\end{equation} 
Indeed, one can easily check that \eqref{Goal} is equal to \eqref{ReorderedSUM} after relabelling all the indices. First, one can check that \eqref{Step1} can be relabelled in the following way:
\begin{equation}
\label{Step2}
R_0(z)+\sum^{\infty}_{k=1} \sum_{\underset{\text{If~} j_{l'} = 1 \text{~then~} j_{l'+1}\neq 2}{(j_1, \dots, j_l) \in \llbracket 1,\dots,6\rrbracket^{l}}} R_0(z) \prod^l_{i=1}\Bigl[ \widetilde{T}^{(j_i)}R_0(z)\Bigr].
\end{equation}  
To build terms of the form $\widetilde{T}^{(7)}$ it is important to collect the terms $\widetilde{T}^{(5)}$ and $\widetilde{T}^{(6)}$ together. We may proceed in the following way: 
\begin{align*}
&\sum^{n}_{k=1} \sum_{{\underset{\text{If~} j_{l'} = 1 \text{~then~} j_{l'+1}\neq 2}{(j_1, \dots, j_l) \in \llbracket 1,\dots,6\rrbracket^{l}}}} R_0(z) \prod^l_{i=1}\Bigl[ \widetilde{T}^{(j_i)}R_0(z)\Bigr]\\
& \quad =  \sum^{n}_{k=1}  \sum_{\underset{m \leq \left \lfloor \frac{k}{2} \right \rfloor}{(s_1, \dots, s_m) \in \{5,6 \}^{m}}} \sum_{\underset{\underset{\forall l' \notin \{l_1, \dots, l_{m}\}, \text{If~} j_{l'} = 1 \text{~then~} j_{l'+1}\neq 2}{ (j_1, \dots, j_{k-2m}) \in \llbracket 1,\dots, 4 \rrbracket^{k-2m}}}{{0=l_0\leq l_1 \leq \dots \leq l_{m} \leq k-2m}}} R_0(z) \prod^{m}_{i=1} \Biggl[\ \prod^{l_i}_{t=l_{i-1}+1}\Bigl[ \widetilde{T}^{(j_t)}R_0(z)\Bigr]\widetilde{T}^{(s_i)} R_0(z)\Biggr]  \prod^{k-2m}_{i=l_m+1}\Bigl[ \widetilde{T}^{(j_i)}R_0(z)\Bigr]  \\
& \qquad  + \sum^{n}_{k=1} \sum_{\underset{\sum^l_{j'=1} \tnop(j_{l'}) > n }{\underset{\text{If~} j_{l'} = 1 \text{~then~} j_{l'+1}\neq 2}{(j_1, \dots, j_l) \in \llbracket 1,\dots,6\rrbracket^{l}}}} R_0(z) \prod^l_{i=1}\Bigl[ \widetilde{T}^{(j_i)}R_0(z)\Bigr].
\end{align*}  

Let us note that: 
\begin{align*}
&\sum^{n}_{k=1}  \sum_{\underset{m \leq \left \lfloor \frac{k}{2} \right \rfloor}{(s_1, \dots, s_m) \in \{5,6 \}^{m}}} \sum_{\underset{\underset{ \forall l' \notin \{l_1, \dots, l_{m}\}, \text{If~} j_{l'} = 1 \text{~then~} j_{l'+1}\neq 2}{  (j_1, \dots, j_{k-2m}) \in \llbracket 1,\dots, 4 \rrbracket^{k-2m}}}{{0=l_0\leq l_1 \leq \dots \leq l_{m} \leq k-2m}}} R_0(z)  \prod^{m}_{i=1} \Biggl[\ \prod^{l_i}_{t=l_{i-1}+1}\Bigl[ \widetilde{T}^{(j_t)}R_0(z)\Bigr]\widetilde{T}^{(s_i)} R_0(z)\Biggr] \prod^{k-2m}_{i=l_m+1}\Bigl[ \widetilde{T}^{(j_i)}R_0(z)\Bigr]\\
& \quad =  \sum^{n}_{k=1} \sum_{\underset{\sum^l_{j'=1} \tnop(j_{l'}) = k}{\underset{\text{If~} j_{l'} = 1 \text{~then~} j_{l'+1}\neq 2}{(j_1, \dots, j_l) \in \left(\llbracket 1,\dots,4\rrbracket\cup\{7\}\right)^{l}}}} R_0(z) \prod^l_{i=1}\Bigl[ \widetilde{T}^{(j_i)}R_0(z)\Bigr].
\end{align*}
Consequently, all the terms of a finite sequences of the series \eqref{Goal} appears once and only once in \eqref{Step1}. As \eqref{Step1} is absolutely convergent as long as $\Re(z)<-25C^2_{\Lambda}$, we can therefore conclude that \eqref{Step1} is also absolutely convergent. To see that these two series converge to the same limit, one can see that: 
\begin{align*}
&\Biggl\|\sum^{n}_{k=1} \sum_{\underset{\text{If~} j_{l'} = 1 \text{~then~} j_{l'+1}\neq 2}{(j_1, \dots, j_l) \in \llbracket 1,\dots,6\rrbracket^{l}}} R_0(z) \prod^l_{i=1}\Bigl[ \widetilde{T}^{(j_i)}R_0(z)\Bigr] - \sum^{n}_{k=1} \sum_{\underset{\sum^l_{j'=1} n'_{\mathrm{op}}(j_{l'}) = k}{\underset{\text{If~} j_{l'} = 1 \text{~then~} j_{l'+1}\neq 2}{(j_1, \dots, j_l) \in \left(\llbracket 1,\dots,4\rrbracket\cup\{7\}\right)^{l}}}} R_0(z) \prod^l_{i=1}\Bigl[ \widetilde{T}^{(j_i)}R_0(z)\Bigr] \Biggr\|\\
&\quad  =  \Biggl\|\sum^{n}_{k=1} \sum_{\underset{\sum^l_{j'=1} \tnop(j_{l'}) > n }{\underset{\text{If~} j_{l'} = 1 \text{~then~} j_{l'+1}\neq 2}{(j_1, \dots, j_l) \in \llbracket 1,\dots,6\rrbracket^{l}}}} R_0(z) \prod^l_{i=1}\Bigl[ \widetilde{T}^{(j_i)}R_0(z)\Bigr]\Biggr\| \\
& \quad \leq  \sum^{\infty}_{k=\left\lfloor \frac{n}{2}\right\rfloor} \sum_{\underset{\text{If~} j_{l'} = 1 \text{~then~} j_{l'+1}\neq 2}{(j_1, \dots, j_l) \in \llbracket 1,\dots,6\rrbracket^{l}}} \bigl\|R_0(z)^{\frac12}\bigr\|^2 \prod^l_{i=1}\Bigl\| R_0(z)^{\frac{1}{2}} \widetilde{T}^{(j_i)}R_0(z)^{\frac{1}{2}}\Bigr\|.
\end{align*}
We can finally use the absolute convergence of \eqref{Step1} to conclude.
\end{proof}

\section{Some Technical Estimates}
The goal of this section is to provide some technical tools required to prove the operator estimates in the next section.
 
\begin{lem}
\label{RegFermionProp}
Let $z\in\mathbb C$ with $\Re(z)<-1$, $ C' \geq C \geq 0$ and $F\in L^2(\mathbb R^d)$. Then, for any  $0\leq \alpha \leq 1$ we have $b(F)\colon \mathscr D(H_0^\alpha)\to \mathscr D(H_0^\alpha)$ and 
\begin{align*}
\left\| (H_0-z+C)^{\alpha} \af(F) R_0(z-C')^{\alpha} \right\| & \leq  2 \| F \|\\
\left\| R_0(z-C')^{\alpha} \cf(F) (H_0-z+C)^{\alpha} \right\| & \leq  2 \| F \|
\end{align*}
\end{lem} 
\begin{proof}
First, recall that $\af(F)$ is a bounded operator with $\|\af(F)\|\leq \|F\|$.
Secondly, compute for $F\in L^2(\mathbb R^d)$ with compact support and $\Phi\in \mathscr H$ with only finitely many particles
\begin{align*}
 (H_0-z+C) \af(F) R_0(z-C')\Phi & =  \af(F) (H_0-z+C) R_0(z-C')\Phi + [H_0,\af(F)] R_0(z-C')\Phi\\
 &= \af(F) (H_0-z+C) R_0(z-C')\Phi  + \af(\wf F) R_0(z-C')\Phi   
\end{align*}
and estimate, using the pull-through formula and Cauchy-Schwarz
\begin{align*}
\|\af(\wf F) R_0(z-C')\Phi\| & \leq \int \wf(k) |F(k)|\| \af(k)R_0(z-C')\Phi\| dk\\
 & \leq  \int \wf(k) |F(k)|\|(H_0-z + \wf(k)+C')^{-\frac12} \af(k)R_0(z-C')^{\frac12}\Phi\| dk\\
 & \leq \int |F(k)| \| \wf(k)^\frac12 \af(k)R_0(z-C')^{\frac12}\Phi\| dk\\
 &\leq \|F\| \|\Phi\|.
\end{align*}
This implies, by a density argument, that for all $F\in L^2(\mathbb R^d)$, we have $b(F)\colon \mathscr D(H_0)\to \mathscr D(H_0)$ and
\[
\|(H_0-z+C)\af(F)R_0(z-C')\|\leq 2\|F\|.
\]

Let $\Psi\in \mathscr D(H_0)$, $\Phi\in \mathscr H$ and define for $\zeta\in\mathbb C$ with $\Re(\zeta) \in [0,1]$ the function
\[
f(\zeta) = \bigl\langle (H_0-\bar{z}+C)^{\bar{\zeta}}\Psi, \af(F) (H_0-z+C')^{-\zeta}\Phi\bigr\rangle. 
\]
This function is continuous and bounded in $\zeta$ with $\Re(\zeta)\in [0,1]$, holomorphic in $\zeta$ with $\Re(\zeta)\in (0,1)$. By the initial estimates established in this proof, we furthermore have for $\lambda\in\mathbb R$:
\[
|f(i\lambda)| \leq \|F\| \|\Psi\|\|\Phi\| \qquad \textup{and}\qquad |f(1+i\lambda)|\leq 2\|F\|  \|\Psi\|\|\Phi\|.
\]
By Hadamard's Three-line Theorem \cite{RS}, we may now conclude the lemma.
\end{proof}

 \begin{lem}
 \label{EstiamtesofAsharp}
 Let $\Re(z)<-1$ and $C, C'\geq 0$, then for any square integrable function $F$, for any positive real number $\delta$ and $\gamma$ such that $\frac12\leq\gamma+\delta \leq 1$ we have for any $\Psi \in \mathscr{H}$:
 \begin{align*}
 \int \wb(q) \bigl[\wb(q)+|\Re(z)|\bigr]^{2(\delta+\gamma)-1} \Bigl\| R_0\bigl(z-\wb(q)-C\bigr)^{\delta} \ab(q) \af(F) R_0\bigl(z-C'\bigr)^{\gamma}\Psi\Bigr\|^2  dq & \leq  4 \bigl\|F\bigr\|^2 \bigl\| \Psi \bigr\|^2, \\
  \int \wb(q) \bigl[\wb(q)+|\Re(z)|\bigr]^{2(\delta+\gamma)-1} \Bigl\| R_0\bigl(z-\wb(q)-C\bigr)^{\delta} \cf(F) \ab(q)  R_0\bigl(z-C'\bigr)^{\gamma}\Psi\Bigr\|^2 dq  & \leq  4 \bigl\|F\bigr\|^2 \bigl\| \Psi \bigr\|^2, \\
   \int \wf(p) \bigl[\wb(p)+|\Re(z)|\bigr]^{2(\delta+\gamma)-1} \Bigl\| R_0\bigl(z-\wf(p)-C\bigr)^{\delta} \af(p) \af(F) R_0\bigl(z-C'\bigr)^{\gamma}\Psi \Bigr\|^2 dp  & \leq  4\bigl\|F\bigr\|^2 \bigl\| \Psi \bigr\|^2, \\
  \int \wf(p) \bigl[\wf(p)+|\Re(z)|\bigr]^{2(\delta+\gamma)-1} \Bigl\| R_0\bigl(z-\wf(p)-C\bigr)^{\delta} \cf(F) \af(p)  R_0\bigl(z-C'\bigr)^{\gamma} \Psi \Bigr\|^2 dp  & \leq  4\bigl\|F\bigr\|^2 \bigl\| \Psi \bigr\|^2. 
  \end{align*}
 \end{lem}
 \begin{proof} Let $\Psi\in\mathscr H$.
 We prove only the first assertion, the other ones being analogous. We therefore have:
 \begin{align*}
  R_0\bigl(z-\wb(q)-C\bigr)^{\delta} \ab(q) \af(F) R_0\bigl(z-C'\bigr)^{\gamma} \Psi 
  &=R_0\bigl(z-\wb(q)-C\bigr)^{\delta} \ab(q) R_0\bigl(z-C'\bigr)^{\gamma} \bigl(H_0-z+C'\bigr)^{\gamma} \af(F) R_0\bigl(z-C'\bigr)^{\gamma} \Psi \\
  & =  R_0\bigl(z-\wb(q)-C\bigr)^{\delta} R_0\bigl(z-\wb(q)-C'\bigr)^{\gamma-\frac12}\ab(q) R_0\bigl(z-C'\bigr)^{\frac12}\\
  &\quad\bigl(H_0-z+C'\bigr)^{\gamma} \af(F) R_0\bigl(z-C'\bigr)^{\gamma} \Psi 
 \end{align*}
 So that: 
  \begin{equation*}
 \Bigl\| R_0\bigl(z-\wb(q)-C\bigr)^{\delta} \ab(q) \af(F) R_0\bigl(z-C\bigr)^{\gamma} \Psi \Bigr\| 
  \leq  \frac{1}{\bigl[\wb(q)+|\Re(z)|\bigr]^{\gamma+\delta-\frac12}}
  \Bigl\|  \ab(q) R_0\bigl(z\bigr)^{\frac12} \bigl(H_0-z+C'\bigr)^{\gamma} \af(F) R_0\bigl(z-C'\bigr)^{\gamma} \Psi\Bigr\|, \\
 \end{equation*}
from which we deduce that: 
\begin{align*}
 &  \int \wb(q) \bigl[\wb(q)+|\Re(z)|\bigr]^{2(\delta+\gamma)-1} \Bigl\| R_0\bigl(z-\wb(q)-C\bigr)^{\delta} \ab(q) \af(F) R_0\bigl(z-C'\bigr)^{\gamma}\Bigr\|^2  dq   \\
& \quad \leq   \Bigl\|\bigl(H_0-z-C'\bigr)^{\gamma} \af(F) R_0\bigl(z-C'\bigr)^{\gamma} \Psi\Bigr\|^2.
\end{align*}
We conclude using Lemma \ref{RegFermionProp}
 \end{proof}
 In the same way, the following Lemma can be proved:
\begin{lem}
\label{EstiamtesofAsharp1}
Let $\Re(z)<$-1 and $C', C \geq 0$, then for any positive real number $\delta$ and $\gamma$ such that $\frac12\leq\gamma+\delta \leq 1$ we have for any $\Psi \in \mathscr{H}$:
\begin{align*}
\int  \wb(q) \bigl[\wb(q) + |\Re(z)|\bigr]^{2 (\delta +\gamma)-1} \Bigl\| R_0\bigl(z-\wb(q)-C\bigr)^{\delta} \ab(q) R_0\bigl(z-C'\bigr)^{\gamma} \Psi \Bigr\|^2 dq & \leq  \bigl\| \Psi \bigr\|^2,\\
\int  \wf(k) \bigl[\wf(k) + |\Re(z)|\bigr]^{2 (\delta +\gamma)-1} \Bigl\| R_0\bigl(z-\wf(k)-C\bigr)^{\delta} \af(k) R_0\bigl(z-C'\bigr)^{\gamma} \Psi \Bigr\|^2 dq & \leq  \bigl\| \Psi \bigr\|^2.
\end{align*}
\end{lem}

 \begin{lem}
 \label{RegTermAlone}
 Let us assume that $\Re(z)<-1$, $C,C'\geq 0 $ and $\delta, \gamma \in [0,1]$ with $\delta+\gamma \geq\frac12$. Then: 
 \begin{align*}
 \Bigl\| R_0\bigl(z-\wb(q)-C\bigr)^{\delta} \af\bigl(F(.,q)\bigr) R_0\bigl(z-C'\bigr)^{\gamma}  \Bigr\|  \leq &  \biggl(\int \frac{|F(k,q)|^2}{\wf(k)\bigl[\wf(k)+|\Re(z)|\bigr]^{2(\delta+\gamma)-1}} dk\biggr)^{\frac12}\\
 & +   \biggl(\int \frac{|F(k,q)|^2}{\bigl[\wb(q)+|\Re(z)|\bigr]^{2\delta}\bigl[\wb(q)+\wf(k)+|\Re(z)|\bigr]^{2\gamma}}   dk\biggr)^{\frac12}.
  \end{align*}
   \end{lem}
 \begin{proof} 
 It is tempting to try to use Hadamard's Three-line Theorem, as in the proof of Lemma~\ref{RegFermionProp}, but that does not seem feasible due to the way the exponents enter into the constants on the right-hand side. Instead, we use a more cumbersome direct estimate that in fact could also be used to give an alternative proof of Lemma~\ref{RegFermionProp}.
 
  First, let us consider the case $C=C'=0$ and compute
 \begin{equation}
 \label{TotEstiLB4}
  R_0\bigl(z-\wb(q)\bigr)^{\delta} \af\bigl(F(.,q)\bigr) R_0\bigl(z\bigr)^{\gamma}   =  \af\bigl(F(.,q)\bigr) R_0\bigl(z-\wb(q)\bigr)^{\delta} R_0\bigl(z\bigr)^{\gamma} +  \bigl[ R_0\bigl(z-\wb(q)\bigr)^{\delta}, \af\bigl(F(.,q)\bigr)\bigr] R_0\bigl(z\bigr)^{\gamma}.  \end{equation}
 Let $\Psi,\Phi\in\mathscr H$. As for the first term on the right-hand side of \eqref{TotEstiLB4}, we estimate:
 \begin{align*}
\Bigl|\Bigl\langle \Phi \Big| \af\bigl(F(.,q)\bigr) R_0\bigl(z-\wb(q)\bigr)^{\delta} R_0\bigl(z\bigr)^{\gamma} \Psi\Bigr\rangle\Bigr| & \leq  \int \bigl|F(k,q)\bigr|\, \Bigl|\Bigl\langle \Phi \Big| \af(k) R_0\bigl(z-\wb(q)\bigr)^{\delta} R_0\bigl(z\bigr)^{\gamma} \Psi\Bigr\rangle\Bigr| dk \\
&\leq \int \frac{|F(k,q)|}{\bigl[\wf(k)+|\Re(z')|\bigr]^{\delta+\gamma - \frac12}} \bigl\| \Phi \bigr\|  \Bigl\| \af(k)  R_0\bigl(z\bigr)^{\frac12} \Psi\Bigr\| dk \\
& \leq  \biggl(\int \frac{|F(k,q)|^2}{\wf(k)\bigl[\wf(k)+|\Re(z)|\bigr]^{2(\delta+\gamma)-1}} dk\biggr)^{\frac12} \bigl\|\Phi \bigr\| \bigl\| \Psi \bigr\|.
 \end{align*}
As for the second term on the right-hand side of \eqref{TotEstiLB4}, we compute:
 \begin{align*}
 &\bigl[ R_0\bigl(z-\wb(q)\bigr)^{\delta}, \af\bigl(F(.,q)\bigr)\bigr] R_0\bigl(z\bigr)^{\gamma} \\ & \quad = R_0\bigl(z-\wb(q)\bigr)^{\delta}\int F(k,q) \Bigl\{1-\Bigl(1 - \wf(k)R_0\bigl(z-\wb(q)-\wf(k)\bigr)\Bigr)^{\delta}   \Bigr\} \af(k) dk  R_0\bigl(z'\bigr)^{\gamma}.
 \end{align*}
 For $n\geq 0$ an integer, let $C_n(\delta) = (-1)^n \binom{\delta}{n} \geq  0$. Then 
\[
1-\Bigl(1 - \wf(k)R_0\bigl(z-\wb(q)-\wf(k)\bigr)\Bigr)^{\delta}  = \sum^{\infty}_{n=1} C_n(\delta) \wf(k)^n R_0\bigl(z-\wb(q)-\wf(k)\bigr)^n
\] 
 and therefore: 
  \begin{align*}
& \Bigl|\Bigl\langle \Phi \Big| \bigl[ R_0\bigl(z-\wb(q)\bigr)^{\delta}, \af\bigl(F(.,q)\bigr)\bigr] R_0\bigl(z\bigr)^{\gamma} \Psi \Bigr\rangle \Bigr| \\
& \quad \leq  \int \bigl|F(k,q)\bigl| \sum^{\infty}_{n=1} C_n(\delta) \wf(k)^n \bigl\| \Phi \bigr\| \Bigl\| R_0\bigl(z-\wb(q)\bigr)^{\delta}  R_0\bigl(z-\wf(k)-\wb(q)\bigr)^n \af(k)  R_0\bigl(z\bigr)^{\gamma} \Psi \Bigr\| \\
& \quad \leq 
\int \frac{|F(k,q)|}{\bigl[\wb(q)+|\Re(z)|\bigr]^{\delta}} \sum^{\infty}_{n=1} C_n(\delta) \frac{\wf(k)^{n-\frac12}}{\bigl[\wf(k) + \wb(q) + |\Re(z)|\bigr]^{n-\frac12+\gamma}} \bigl\| \Phi \bigr\| \Bigl\| \wf(k)^{\frac12} \af(k)   R_0\bigl(z\bigr)^{\frac12} \Psi \Bigr\| dk \\
& \quad  \leq  
\int \frac{|F(k,q)|}{\bigl[\wb(q)+|\Re(z)|\bigr]^{\delta}
\bigl[\wb(q)+\wf(k)+|\Re(z)|\bigr]^{\gamma}}  \bigl\| \Phi \bigr\| \Bigl\| \wf(k)^{\frac12} \af(k)   R_0\bigl(z\bigr)^{\frac12} \Psi \Bigr\| dk \\
 &\quad \leq 
\biggl(\int \frac{|F(k,q)|^2}{\bigl[\wb(q)+|\Re(z)|\bigr]^{2\delta}\bigl[\wb(q)+\wf(k)+|\Re(z)|\bigr]^{2\gamma}}   dk\biggr)^{\frac12} \bigl\| \Phi \bigr\| \bigl\|  \Psi \bigr\|.
 \end{align*}
 
To conclude the proof, we note that:
\begin{equation*}
\Bigl\| R_0\bigl(z-\wb(q)-C\bigr)^{\delta} \af\bigl(F(.,q)\bigr) R_0\bigl(z-C'\bigr)^{\gamma}  \Bigr\|  \leq   \Bigl\| R_0\bigl(z-\wb(q)\bigr)^{\delta} \af\bigl(F(.,q)\bigr)R_0\bigl(z\bigr)^{\gamma} \bigl(H_0-z\bigr)^{\gamma} R_0\bigl(z-C'\bigr)^{\gamma}  \Bigr\|.
\end{equation*}
  \end{proof}

\begin{lem}
\label{EstimateWtihRespectToOtherPowers}
Let $F$ be a square integrable function  of the form $F(k,q)=h(k,q)g(k\pm q)$ where $h$ is a bounded function and $g$ fulfils  Hypothesis~\ref{MainHypothesis}. Then, For any two sets of exponents $\alpha, \beta, \gamma, \delta \geq 0$, and $\alpha', \beta', \gamma', \delta' \geq 0$ with 
\begin{equation*}
\alpha' + \gamma'  =  \alpha + \gamma \qquad \textup{and} \qquad
\beta' + \delta' =  \beta + \delta,
\end{equation*}
 there exists a constant $C_s>0$, depending only on $\mb,\mf$, such that we have 
\begin{align*}
\forall \lambda\geq 0:\quad &\int\frac{|F(k,q)|^2}{\wb(q)^{\alpha}\bigl[\wb(q)+\lambda\bigr]^{\beta}\wf(k)^{\gamma}\bigl[\wf(k)+\lambda\bigr]^{\delta}} dk dq\\
&\qquad \leq C_s^{\alpha'-\alpha + \beta'-\beta}  \int\frac{|F(k,q)|^2}{\wb(q)^{\alpha'}\bigl[\wb(q)+\lambda\bigr]^{\beta'}\wf(k)^{\gamma'}\bigl[\wf(k)+\lambda\bigr]^{\delta'}} dk dq.
\end{align*}
\end{lem}

\begin{proof}
Due to the symmetry between the two variables, $q$ and $k$, we may assume that $\alpha\leq  \alpha'$ and $\beta\leq \beta'$. By assumption, we can assume that the function $F$ is of the form $F(k,q) = g(k\pm q) h(k,q)$
and we may assume without loss of generality that 
\[
F(k,q) = g(k- q) h(k,q),
\]
the other case being similar. Observe that 
\begin{align*}
    \int\frac{|F(k,q)|^2}{\wb(q)^{\alpha}\bigl[\wb(q)+\lambda\bigr]^{\beta}\wf(k)^{\gamma}\bigl[\wf(k)+\lambda\bigr]^{\delta}} dk dq  & = \int\frac{|g(k- q)|^2 |h(k,q)|^2}{\wb(q)^{\alpha}\bigl[\wb(q)+\lambda\bigr]^{\beta}\wf(k)^{\gamma}\bigl[\wf(k)+\lambda\bigr]^{\delta}} dk dq \\ 
&    = \int \frac{|g(v)|^2 |h(k ,v+k)|^2}{\wb(v+k)^{\alpha}\bigl[\wb(v+k)+\lambda\bigr]^{\beta}\wf(k)^{\gamma}\bigl[\wf(k)+\lambda\bigr]^{\delta}}  dv dk.
   \end{align*}
 First, if $|k| \geq 2$ and $|v|\leq 1$, we have: 
\begin{align*}
\wb(v+k)^2 & \geq  |k|^2 + |v|^2 - 2|k| |v| + \mb^2 =  \bigl(|k|^2+\mb^2\bigr)\Bigl( 1 + \frac{|v|^2}{\wb(k)^2} - 2\frac{|k| |v|}{ \wb(k)^2} \Bigr)\\
& \geq  \wb(k)^2 \Bigl( 1 -  \frac{|v|}{\wb(k)} \Bigr)^2 \geq  \wb(k)^2 \frac{1}{4} \geq \frac14 \min\Bigl\{1,\frac{\mb^2}{\mf^2}\Bigr\} \wf(k)^2.
\end{align*}
If, on the other hand, $|k|\leq 2$, we have
\begin{align*}
    \wf(k)^2 = k^2 +\mf^2 \leq 4 + \mf^2 = \frac{4+\mf^2}{\mb^2}\, \mb^2\leq \frac{4+\mf^2}{\mb^2}\, \wb(v+k)^2.
\end{align*}

Taken together, we conclude the existence of a $c>0$, depending only on $\mb$ and $\mf$, such that for all $\lambda\geq 0$ and $k,v\in\mathbb R^d$ with $\|v\|\leq 1$:
\[
\frac1{\wb(v+k)+\lambda} \leq \frac{c}{\wf(k)+\lambda}.
\]
This implies that: 
\begin{align*}
& \frac{|g(v)|^2 |h(k ,v+k)|^2}{\wb(v+k)^{\alpha}\bigl[\wb(v+k)+\lambda\bigr]^{\beta}\wf(k)^{\gamma}\bigl[\wf(k)+\lambda\bigr]^{\delta}}   \\
 & \qquad \leq c^{\alpha'-\alpha + \beta'-\beta} \frac{|g(v)|^2 |h(k ,v+k)|^2}{\wb(v+k)^{\alpha'}\bigl[\wb(v+k)+\lambda\bigr]^{\beta'} \wf(k)^{\gamma'}\bigl[\wf(k)+\lambda\bigr]^{\delta'}}  
\end{align*}
and completes the proof.
\end{proof}

\section{Operator Estimates}
\label{InitialisationEstimates}

In this section, we will assume that the functions $F, G, F_1, F_2, F_3$ fulfil the hypothesis of Lemma~\ref{EstimateWtihRespectToOtherPowers}. 


\begin{lem}
\label{IE1}
Let us assume that $\Re(z)<-1$, $\frac12\leq  \beta \leq 1$. Then, 
\[
\begin{aligned}
\bigl\|H^{\ab \af}(F) R_0(z)^{\beta}\bigr\| & \leq  K^{(1)}_{z,\beta}(F), & \qquad 
\bigl\|R_0(z)^{\beta} H^{\cb \cf}(F) \bigr\| & \leq  K^{(1)}_{z,\beta}(F),\\
\bigl\|  H^{\ab \cf}(F) R_0(z)^{\beta}\bigr\| & \leq  K^{(1)}_{z,\beta}(F), & \qquad
\bigl\| R_0(z)^{\beta} H^{\cb \af}(F) \bigr\| & \leq  K^{(1)}_{z,\beta}(F).
\end{aligned}
\]
where
\begin{equation*}
K^{(1)}_{z,\beta}(F) = \left(\int \frac{|F(k,q)|^2 }{[\wf(k)+|\Re(z)|]^{2 \beta -1 } \wf(k)} dk dq\right)^{\frac12}. 
\end{equation*}
\end{lem}

\begin{proof}
We prove only the first estimate as the other ones are analogous. Let $\Psi,\Phi\in\mathscr{H}$. We then have:
\begin{align*}
\Bigl| \bigl\langle \Psi \big| H^{\ab \af}(F) R_0(z)^{\beta}  \Phi \bigr\rangle \Bigr|   & =  \biggl| \int  \bigl\langle \Psi \big|  \af\bigl(F(.,q)\bigr) \ab(q) R_0(z)^{\beta}  \Phi \bigr\rangle  dq  \biggr|\\
& \leq   \int  \Bigl|\bigl\langle \Psi \big|  \af\bigl(F(.,q)\bigr) \ab(q) R_0(z)^{\beta}  \Phi \bigr\rangle\Bigr|  dq  \\
& \leq  \int  \bigl\| \Psi \bigr\| \Bigl\|  \af\bigl(F(.,q)\bigr) \ab(q) R_0(z)^{\beta}  \Phi \Bigr\|  dq \\
& \leq  \bigl\| \Psi \bigr\|  \int \frac{\| F(.,q) \|}{\wb(q)^{\frac12}\bigl[\wb(q)+|\Re(z)|\bigr]^{\beta-\frac12}} \Bigl\| \wb(q)^{\frac12}   \ab(q) R_0(z)^{\frac12}  \Phi \Bigr\|  dq\\
&\leq  \bigl\| \Psi \bigr\|  \biggl(\int \frac{\| F(.,q)\|^2}{\wb(q)\bigl[\wb(q)+|\Re(z)|\bigr]^{2\beta-1}}dq\biggr)^{\frac12} \biggl(\int \Bigl\| \wb(q)^{\frac12}   \ab(q) R_0(z)^{\frac12}  \Phi \Bigr\|^2  dq\biggr)^{\frac12}\\
& \leq  \bigl\| \Phi \bigr\| \bigl\| \Psi \bigr\|  \biggl(\int \frac{\| F(.,q) \|^2}{\wb(q)\bigl[\wb(q)+|\Re(z)|\bigr]^{2\beta-1}}dq\biggr)^{\frac12} .
\end{align*}
We conclude the proof by appealing to Lemma~\ref{EstimateWtihRespectToOtherPowers}. 
\end{proof}

\begin{lem}
\label{IE2}
Let us assume that $\Re(z) < -1$, $0 \leq \gamma \leq 1$ and $0 \leq \delta \leq 1$. Then, 
\[
\Bigl\|R^{\gamma}_0(z) \bigl( H^{\ab\af}(F) R_0(z) H^{\cb \cf}(G) + E(F,G) \bigr)R^{\delta}_0(z)\Bigr\| \leq \left(5+C_s^{1+\gamma+\delta}\right) K^{(2)}_{z, \gamma+\delta}(F,G),
\]
with 
\begin{equation}\label{EFG}
E(F,G) = -\int \frac{\overline{F(k,q)}G(k,q)}{\wf(k) + \wb(q)} dk dq,
\end{equation}
\begin{equation}
\label{K2}
K^{(2)}_{z, \gamma+\delta}(F,G)= A_{z,\delta+\gamma}(F)A_{z,\delta+ \gamma}(G)
\end{equation}
and
\[
A_{z,\delta+\gamma}(F) = \biggl(\int \frac{|F(k,q)|^2}{\wb(q) \bigl[\wb(q) + |\Re(z)|\bigr]^{\gamma + \delta}} dk dq\biggr)^{\frac12} 
\]
\end{lem}

\begin{proof}
Let $\Phi,\Psi\in\mathscr{H}$. First, we have that:
\begin{align}\label{C2-Step1}
\nonumber H^{\ab \af}(F) R_0(z) H^{\cb \cf}(G) = & \int \cb(q_2) \af\bigl(F(.,q_1)\bigr) R_0\bigl(z-\wb(q_1) - \wb(q_2)\bigr) \cf\bigl(G(.,q_2)\bigr) \ab(q_1) d q_1 d q_2\\
\nonumber   & - \int F(k_1,q) G(k_2,q) \cf(k_2) R_0\bigl(z-\wb(q) -\wf(k_1) - \wf(k_2) \bigr) \af(k_1)  dk_1 dk_2 dq \\
 & + \int F(k,q) G(k,q) R_0\bigl(z-\wb(q) -\wf(k)  \bigr)   dk dq 
\end{align}
and we will estimate each term separately. For the first term on the right-hand side of \eqref{C2-Step1}: 
\begin{align*}
& \Bigl| \Bigl\langle \Psi \Big| R_0(z)^{\gamma} \int \cb(q_2) \af\bigl(F(.,q_1)\bigr) R_0\bigl(z-\wb(q_1) - \wb(q_2)\bigr) \cf\bigl(G(.,q_2)\bigr) \ab(q_1) d q_1 d q_2  R_0(z)^{\delta} \Phi \Bigr\rangle \Bigr|  \\
& \quad  \leq  \int \Bigl| \Bigl\langle \ab(q_2) R_0\bigl(\bar{z}\bigr)^{\frac12+\frac{\gamma+\delta}{2}}\Psi \Big| \bigl(H_0-z+\wb(q_2)\bigr)^{\frac{1+\delta-\gamma}{2}}  \af\bigl(F(.,q_1)\bigr) R_0\bigl(z-\wb(q_1) - \wb(q_2)\bigr)  \\
& \qquad  \cf\bigl(G(.,q_2)\bigr)\bigl(H_0-z+\wb(q_1)\bigr)^{\frac{1+\gamma-\delta}{2}} \ab(q_1)   R_0\bigl(z\bigr)^{\frac12+\frac{\gamma+\delta}{2}} \Phi \Bigr\rangle  \Bigr| d q_1 d q_2   \\
&  \quad \leq  \int \Bigl\| \ab(q_2) R_0\bigl(\bar{z}\bigr)^{\frac12+\frac{\gamma+\delta}{2}}\Psi \Bigr\| \Bigl\| \bigl(H_0-z+\wb(q_2)\bigr)^{\frac{1+\delta-\gamma}{2}}  \af\bigl(F(.,q_1)\bigr) R_0\bigl(z-\wb(q_1) - \wb(q_2)\bigr)    \\
& \qquad   \cf\bigl(G(.,q_2)\bigr)\bigl(H_0-z+\wb(q_1)\bigr)^{\frac{1+\gamma-\delta}{2}} \ab(q_1)   R_0\bigl(z\bigr)^{\frac12+\frac{\gamma+\delta}{2}} \Phi \Bigr\| d q_1 d q_2.
\end{align*}
Using Lemma \ref{RegFermionProp}, we arrive at:
\begin{align*}
& \Bigl| \Bigl\langle \Psi \Big| R_0\bigl(z\bigr)^{\gamma} \int \cb(q_2) \af\bigl(F(.,q_1)\bigr) R_0\bigl(z-\wb(q_1) - \wb(q_2)\bigr) \cf\bigl(G(.,q_2)\bigr) \ab(q_1) d q_1 d q_2  R_0\bigl(z\bigr)^{\delta} \Phi \Bigr\rangle  \Bigr| \\
 & \quad \leq 4   \int \bigl\| F(.,q_1) \bigr\| \bigl\| G(.,q_2) \bigr\|  \Bigl\| \ab(q_2) R_0\bigl(\bar{z}\bigr)^{\frac12+\frac{\gamma+\delta}{2}}\Psi \Bigr\| \Bigl\|  \ab(q_1)   R_0\bigl(z\bigr)^{\frac12+\frac{\gamma+\delta}{2}} \Phi \Bigr\| d q_1 d q_2 \\
  & \quad \leq 4 \int  \frac{\| F(.,q_1) \|}{\wb(q_1)^{\frac12} \bigl[\wb(q_1) + |\Re(z)|\bigr]^{\frac{\gamma + \delta}{2}}} \frac{\| G(.,q_2) \|}{ \wb(q_2)^{\frac12} \bigl[\wb(q_2) + |\Re(z)|\bigr]^{\frac{\gamma + \delta}{2}}}  \Bigl\| \wb(q_2)^{\frac12} \ab(q_2) R_0\bigl(\bar{z}\bigr)^{\frac12}\Psi \Bigr\|\\
  &\qquad \Bigl\| \wb(q_1)^{\frac12}  \ab(q_1)   R_0\bigl(z\bigr)^{\frac12} \Phi \Bigr\| d q_1 d q_2 \\
    & \quad \leq 4 \biggl( \int \biggl( \frac{\| F(.,q_1) \|}{\wb(q_1)^{\frac12} \bigl[\wb(q_1) + |\Re(z)|\bigr]^{\frac{\gamma + \delta}{2}}} \frac{\| G(.,q_2) \|}{ \wb(q_2)^{\frac12} \bigl[\wb(q_2) + |\Re(z)|\bigr]^{\frac{\gamma + \delta}{2}}}  \biggr)^2 d q_1 d q_2\biggr)^{\frac12}\\
  &\qquad  \biggl(\int \Bigl\| \wb(q_2)^{\frac12} \ab(q_2) R_0\bigl(\bar{z}\bigr)^{\frac12}\Psi \Bigr\|^2 \Bigl\| \wb(q_1)^{\frac12}  \ab(q_1)   R_0\bigl(z\bigr)^{\frac12} \Phi \Bigr\|^2 d q_1 d q_2 \biggr)^{\frac12}\\
 & \quad \leq 4   \biggl(\int \frac{|F(k,q)|^2}{\wb(q) \bigl[\wb(q) + |\Re(z)|\bigr]^{\gamma + \delta}} dk dq\biggr)^{\frac12} \biggl(\int \frac{|G(k,q)|^2}{\wb(q) \bigl[\wb(q) + |\Re(z)|\bigr]^{\gamma + \delta}} dk dq\biggr)^{\frac12} \bigl\| \Psi \bigr\| \bigl\| \Phi \bigr\|.
\end{align*}
As for the second term on the right-hand side of \eqref{C2-Step1}, we have: 
\begin{align*}
& \Bigl| \Bigl\langle \Psi \Big| R_0\bigl(z\bigr)^{\gamma} \int  \overline{F(k_1,q)} G(k_2,q) \cf(k_2) R_0\bigl(z-\wb(q) - \wf(k_1)-\wf(k_2)\bigr) \af(k_1)  d k_1 d k_2 d q  R_0\bigl(z\bigr)^{\delta} \Phi \Bigr\rangle\Bigr|  \\
& \quad \leq   \int \biggl( \int\bigl|  F(k_1,q) G(k_2,q)\bigr| dq \biggr) \Bigl\| \af(k_2) R_0\bigl(\bar{z}\bigr)^{\frac{1+\gamma+\delta}{2}} \Psi \Bigr\| \Bigl\|     \af(k_1)   R_0\bigl(z\bigr)^{\frac{1+\gamma+\delta}{2}} \Phi \Bigr\| d k_1 d k_2   \\
& \quad \leq   \Biggl(\int \frac{\Bigl(\int\bigl|   F(k_1,q) G(k_2,q)\bigr| dq \Bigr)^2}{\wf(k_1)\bigl[\wf(k_1)+|\Re(z)|\bigr]^{\gamma+\delta} \wf(k_2)\bigl[\wf(k_2)+|\Re(z)|\bigr]^{\gamma+\delta}} d k_1 d k_2 \Biggr)^{\frac12}\bigl\|  \Psi \bigr\| \bigl\|      \Phi \bigr\|  \\
& \quad \leq   \biggl( \int\frac{|F(k,q)|^2}{\wf(k) \bigl[\wf(k) + |\Re(z)|\bigr]^{\delta + \gamma}} d k d q \biggr)^{\frac12} \biggl( \int \frac{|G(k,q)|^2}{\wf(k) \bigl[\wf(k) + |\Re(z)|\bigr]^{\delta + \gamma}} d k d q \biggr)^{\frac12} \bigl\|  \Psi \bigr\| \bigl\| \Phi \bigr\|.
\end{align*}
By using Lemma \ref{EstimateWtihRespectToOtherPowers} we can conclude that 
\begin{align*}
& \Bigl| \Bigl\langle \Psi \Big| R_0\bigl(z\bigr)^{\gamma} \int  \overline{F(k_1,q)} G(k_2,q) \cf(k_2) R_0\bigl(z-\wb(q) - \wf(k_1)-\wf(k_2)\bigr) \af(k_1)  d k_1 d k_2 d q  R_0\bigl(z\bigr)^{\delta} \Phi \Bigr\rangle\Bigr|  \\
& \quad \leq   C_s^{1+\delta+\gamma}\biggl( \int\frac{|F(k,q)|^2}{\wf(k) \bigl[\wf(k) + |\Re(z)|\bigr]^{\delta + \gamma}} d k d q \biggr)^{\frac12} \biggl( \int \frac{|G(k,q)|^2}{\wf(k) \bigl[\wf(k) + |\Re(z)|\bigr]^{\delta + \gamma}} d k d q \biggr)^{\frac12} \bigl\|  \Psi \bigr\| \bigl\| \Phi \bigr\|.
\end{align*}

Finally, we estimate the difference between the third term on the right-hand side of \eqref{C2-Step1} and the counter-term $E(F,G)$:
\begin{align*}
&\biggl| \biggl\langle \Psi \bigg| R_0\bigl(z\bigr)^{\gamma}\biggl\{ \int\overline{F(k,q)}  G(k,q) R_0\bigl(z-\wb(q) -\wf(k)  \bigr)   dk dq - \int \frac{\overline{F(k,q)} G(k,q)}{\wb(q)+\wf(k)}d q d k \biggr\} R_0\bigl(z\bigr)^{\delta} \Phi \bigg\rangle  \biggr|  \\
& \quad \leq   \int \bigl|F(k,q) G(k,q)\bigr| \biggl| \biggl\langle \Psi \bigg| R_0\bigl(z\bigr)^{\gamma}\biggl\{  R_0\bigl(z-\wb(q) -\wf(k)  \bigr)   - \frac{1}{\wb(q)+\wf(k)} \biggr\} R_0\bigl(z\bigr)^{\delta} \Phi \biggr\rangle  \biggr| dk dq   \\
& \quad =   \int \bigl|F(k,q) G(k,q)\bigr| \biggl| \biggl\langle \Psi \bigg| R_0\bigl(z\bigr)^{\gamma}\biggl\{      \frac{(H_0-z) R_0(z-\wb(q) -\wf(k)  )}{\wb(q)+\wf(k)} \biggr\} R_0\bigl(z\bigr)^{\delta} \Phi \biggr\rangle  \biggr| dk dq  \\
& \quad \leq   \int \bigl|F(k,q) G(k,q)\bigr| \bigl\|  \Psi \bigr\| \biggl\| R_0\bigl(z\bigr)^{\gamma}\biggl\{ \frac{(H_0-z) R_0(z-\wb(q) -\wf(k)  )}{\wb(q)+\wf(k)} \biggr\} R_0\bigl(z\bigr)^{\delta} \Phi \biggr\| dk dq  \\
& \quad \leq  \int \bigl|F(k,q) G(k,q)\bigr|  \bigl\| \Psi \bigr\| \biggl\| \biggl\{ \frac{ R_0(z-\wb(q) -\wf(k)  )^{\gamma+\delta}}{\wb(q)+\wf(k)} \biggr\} \Phi \biggr\| dk dq   \\
& \quad \leq  \int \frac{|F(k,q) G(k,q)|}{\bigl[\wb(q)+\wf(k)\bigr] \bigl[\wb(q)+\wf(k)+|\Re(z)|\bigr]^{\gamma + \delta}} dk dq  \bigl\| \Psi \bigr\| \bigl\| \Phi \bigr\|  \\
& \quad \leq  \biggl(\int \frac{|F(k,q)|^2}{\bigl[\wb(q)+\wf(k)\bigr] \bigl[\wb(q)+\wf(k)+|\Re(z)|\bigr]^{\gamma + \delta}} dk dq \biggr)^{\frac12} \\ & \qquad \biggl(\int \frac{|G(k,q)|^2}{\bigl[\wb(q)+\wf(k)\bigr] \bigl[\wb(q)+\wf(k)+|\Re(z)|\bigr]^{\gamma + \delta}} dk dq \biggr)^{\frac12}  \bigl\| \Psi \bigr\| \bigl\| \Phi \bigr\|. \\
\end{align*}
This completes the proof of the lemma.
\end{proof}


\begin{lem}
\label{IE3}
Let us assume that $\Re(z)<-1$ and $0\leq \beta \leq 1$. We then have: 
\begin{align*}
\bigl\| H^{\ab \af}(F) R_0(z) H^{\cb \af}(G) R_0(z)^{\beta}  \bigr\| & \leq 4 ( 4 + C_s^{1+\beta})  K^{(2)}_{z,\beta}(F,G)    \\
\bigl\| R_0(z)^{\beta} H^{\ab \cf}(G) R_0(z) H^{\cb \cf}(F)   \bigr\| & \leq 4 ( 4 + C_s^{1+\beta})  K^{(2)}_{z,\beta}(F,G) 
\end{align*}
where $K^{(2)}_{z,\beta}$ is defined by \eqref{K2}.
\end{lem}
\begin{proof}
Let $\Psi,\Phi \in \mathscr{H}$. We will treat only the first case as the second one is the adjoint case. First, 
\begin{align*}
H^{\ab \af}(F) R_0(z) H^{\cb \af}(G) & =  \int \cb(q_2) \af(F(.,q_1)) R_0(z-\wb(q_1) - \wb(q_2)) \af(G(.,q_2)) \ab(q_1)  d q_1 d q_2\\
& \quad + \int  \af(F(.,q)) R_0(z-\wb(q)) \af(G(.,q))   d q.
\end{align*}
On the one hand, 
\begin{align*}
& \left| \left< \Psi \left|  \int \cb(q_2) \af(F(.,q_1)) R_0(z-\wb(q_1) - \wb(q_2)) \af(G(.,q_2)) \ab(q_1) d q_1 d q_2 R_0(z)^{\beta}\Phi \right> \right| \right. \\ 
& \quad\leq  \int \left\| R_0(z-\wb(q_2))^{\frac{1+\beta}{2}} \cf(F(.,q_1))\ab(q_2)  \Psi \right\| \left\|    R_0(z- \wb(q_1))^{\frac{1-\beta}{2}} \af(G(.,q_2)) \ab(q_1)    R_0(z)^{\beta}\Phi \right\| dq_1 d q_2  \\
& \quad\leq  \int \left( \int \left( \wb(q_2) [\wb(q_2) + |\Re(z)|]^{\beta}\right)\left\| R_0(z-\wb(q_2))^{\frac{1+\beta}{2}} \cf(F(.,q_1))\ab(q_2)  \Psi \right\|^2 d q_2 \right)^{\frac12}  \\
&  \qquad \biggl(\int \frac{1}{\left( \wb(q_2) [\wb(q_2) + |\Re(z)|]^{\beta}\right)}\left\|    R_0(z- \wb(q_1))^{\frac{1-\beta}{2}} \af(G(.,q_2)) \ab(q_1)    R_0(z)^{\beta}\Phi \right\|^2 dq_2 \biggr)^{\frac12} dq_1   \\
\end{align*}
and using Lemma \ref{EstiamtesofAsharp} we have:  
\begin{align*}
&\left| \left< \Psi \left|  \int \cb(q_2) \af(F(.,q_1)) R_0(z-\wb(q_1) - \wb(q_2)) \af(G(.,q_2)) \ab(q_1) d q_1 d q_2 R_0(z)^{\beta}\Phi \right> \right| \right.\\ 
 & \quad \leq 4  \int \frac{1}{\left( \wb(q_1) [\wb(q_1) + |\Re(z)|]^{\beta}\right)^{\frac12}}\left\| F(.,q_1) \right\| \bigl\| \Psi \bigr\| \\
&\qquad \biggl(\int \frac{\left( \wb(q_1) [\wb(q_1) + |\Re(z)|]^{\beta}\right)}{\left( \wb(q_2) [\wb(q_2) + |\Re(z)|]^{\beta}\right)}\left\|    R_0(z- \wb(q_1))^{\frac{1-\beta}{2}} \af(G(.,q_2)) \ab(q_1)    R_0(z)^{\beta}\Phi \right\|^2 dq_2 \biggr)^{\frac12} dq_1    \\
 & \quad\leq 4 \biggl(\int \frac{|F(k_1,q_1)|^2}{\left( \wb(q_1) [\wb(q_1) + |\Re(z)|]^{\beta}\right)} d k_1 d q_1\biggr)^{\frac12} \bigl\|\Psi\bigr\| \\
& \qquad \biggl(\int \biggl(\int \frac{\left( \wb(q_1) [\wb(q_1) + |\Re(z)|]^{\beta}\right)}{\left( \wb(q_2) [\wb(q_2) + |\Re(z)|]^{\beta}\right)}\left\|    R_0(z- \wb(q_1))^{\frac{1-\beta}{2}} \af(G(.,q_2)) \ab(q_1)    R_0(z)^{\beta}\Phi \right\|^2 dq_1 \biggr) dq_2 \biggr)^{\frac12} \\
 & \quad \leq 16 \biggl(\int \frac{|F(k,q)|^2}{\left( \wb(q) [\wb(q) + |\Re(z)|]^{\beta}\right)} d k d q\biggr)^{\frac12} \biggl(\int \frac{|G(k,q)|^2}{\wb(q) [\wb(q) + |\Re(z)|]^{\beta}} d k d q\biggr)^{\frac12} \bigl\| \Psi \bigr\| \bigl\| \Phi \bigr\|
\end{align*}
On the other hand: 
\begin{align*}
& \left| \left< \Psi \left|  \int  \af(F(.,q)) R_0(z-\wb(q) ) \af(G(.,q))   d q R_0(z)^{\beta}\Phi \right> \right| \right. \\
& \quad \leq \bigl\| \Psi \bigr\| \int \left| F(k_1,q) \right| \left\|     R_0(z-\wf(k_1) -\wb(q)) \af(G(.,q))  \af(k_1) R_0(z)^{\beta}\Phi \right\|  d q d k_1  \\
& \quad \leq  \bigl\| \Psi \bigr\| \int \frac{\left| F(k_1,q) \right|}{[\wb(q)+|\Re(z)|]^{\frac{1+\beta}{2}}} \left\|     R_0(z-\wf(k_1) -\wb(q))^{\frac{1-\beta}{2}} \af(G(.,q))  \af(k_1) R_0(z)^{\beta}\Phi \right\|  d q d k_1 \\
& \quad \leq  \bigl\| \Psi \bigr\| \int \biggl( \int \frac{\left| F(k_1,q) \right|^2}{[\wb(q)+|\Re(z)|]^{1+\beta} \wf(k_1) [\wf(k_1) + |\Re(z)|]^{\beta}} d k_1 \biggr)^{\frac12} \\
&\qquad \biggl(\int \left(\wf(k_1) [\wf(k_1) + |\Re(z)|]^{\beta}\biggr) \left\|     R_0(z-\wf(k_1) -\wb(q))^{\frac{1-\beta}{2}} \af(G(.,q))  \af(k_1) R_0(z)^{\beta}\Phi \right\|^2   d k_1 \right)^{\frac12} d q \\
& \quad \leq  4 \bigl\|\Psi\bigr\| \int \biggl( \int \frac{\left| F(k_1,q) \right|^2}{[\wb(q)+|\Re(z)|]^{1+\beta} \wf(k_1) [\wf(k_1) + |\Re(z)|]^{\beta}} d k_1 \biggr)^{\frac12} \bigl\|G(.,q)\bigr\| \bigl\| \Phi\bigr\| d q \\
& \quad \leq 4 \bigl\| \Phi\bigr\|  \bigl\| \Psi \bigr\|  \biggl( \int \frac{\left| F(k,q) \right|^2}{ \wf(k) [\wf(k) + |\Re(z)|]^{\beta}} d k dq \biggr)^{\frac12} \biggl( \int \frac{\left| G(k,q) \right|^2}{[\wb(q)+|\Re(z)|]^{1+\beta} } d k dq \biggr)^{\frac12},
\end{align*}
where Lemma \ref{EstiamtesofAsharp} has been used.  Lemma \ref{EstimateWtihRespectToOtherPowers} can now be used to conclude that 
\begin{align*}
& \left| \left< \Psi \left|  \int  \af(F(.,q)) R_0(z-\wb(q) ) \af(G(.,q))   d q R_0(z)^{\beta}\Phi \right> \right| \right. \\
& \quad \leq  4 C_s^{1+\beta} \bigl\| \Phi\bigr\|  \bigl\| \Psi \bigr\|  \biggl( \int \frac{\left| F(k,q) \right|^2}{ \wb(q) [\wb(q) + |\Re(z)|]^{\beta}} d k dq \biggr)^{\frac12} \biggl( \int \frac{\left| G(k,q) \right|^2}{\wb(q)[\wb(q)+|\Re(z)|]^{\beta} } d k dq \biggr)^{\frac12}.
\end{align*}
\end{proof}


\begin{lem}
\label{IE4}
Let us assume that $\Re(z)<-1$ and $0 \leq \delta, \gamma \leq 1$. We have: 
\begin{equation*}
\Bigl\| R_0\bigl(z\bigr)^{\delta} H^{\ab \cf}(F) R_0\bigl(z\bigr) H^{\cb \af}(G) R_0\bigl(z\bigr)^{\gamma} \Bigr\| \leq \left(17+2C_s+C_s^2\right) K^{(2)}_{z,\gamma+\delta}(F,G),
\end{equation*}
where $K^{(2)}_{z,\beta}$ is defined by \eqref{K2}.
\end{lem}

\begin{proof} Compute first
\begin{align*}
 H^{\ab \cf}(F) R_0\bigl(z\bigr) H^{\cb \af}(G)  & =  \int  \cf\bigl(F(.,q_1)\bigr) \cb(q_1) R_0\bigl(z-\wb(q_2)-\wb(q_1)\bigr) \ab(q_2) \af\bigl(G(.,q_2)\bigr) d q_1 d q_2 \\
 & \quad + \int   \cf\bigl(F(.,q)\bigr)  R_0\bigl(z- \wb(q)\bigr)  \af\bigl(G(.,q)\bigr) d q.  \\
\end{align*}
Let $\Psi,\Phi\in \mathscr{H}$. We then have:
\begin{align*}
& \int \Bigl| \Bigl\langle \Psi  \Big|   R_0\bigl(z\bigr)^{\delta}  \cf\bigl(F(.,q_1)\bigr) \cb(q_1) R_0\bigl(z-\wb(q_2)-\wb(q_1)\bigr) \ab(q_2) \af\bigl(G(.,q_2)\bigr) R_0\bigl(z\bigr)^{\gamma}  \Phi \Bigr\rangle \Bigr| d q_1 d q_2  \\
& \quad \leq  \int \Bigl\| R_0\bigl(z-\wb(q_2)-\wb(q_1)\bigr)^{\frac{1+\gamma-\delta}{2}}\ab(q_1) \af\bigl(F(.,q_1)\bigr) R_0\bigl(z\bigr)^{\delta}  \Psi \Bigr\| \\
&\qquad \Bigl\|     R_0\bigl(z-\wb(q_2)-\wb(q_1)\bigr)^{\frac{1+\delta-\gamma}{2}} \ab(q_2) \af\bigl(G(.,q_2)\bigr)  R_0\bigl(z\bigr)^{\gamma}  \Phi \Bigr\| d q_1 d q_2 
\end{align*}
and using Lemma \ref{EstiamtesofAsharp} we obtain:
\begin{align*}
    & \int \Bigl| \Bigl\langle \Psi  \Big|   R_0(z)^{\delta}  \cf\bigl(F(.,q_1)\bigr) \cb(q_1) R_0\bigl(z-\wb(q_2)-\wb(q_1)\bigr) \ab(q_2) \af\bigl(G(.,q_2)\bigr) R_0\bigl(z\bigr)^{\gamma}  \Phi \Bigr\rangle \Bigr|  d q_1 d q_2 \leq 16 K^{(2)}_{z,\gamma+\delta}(F,G) .
\end{align*}
On the other hand,
\begin{align*}
&\int  \Bigl| \Bigl\langle \Psi \Big| R_0\bigl(z\bigr)^{\delta} \cf\bigl(F(.,q)\bigr)  R_0\bigl(z- \wb(q)\bigr)  \af\bigl(G(.,q)\bigr) R_0\bigl(z\bigr)^{\gamma} \Phi \Bigr\rangle \Bigr| d q \\
 & \quad \leq  \int  \Bigl\|  R_0\bigl(\bar{z}- \wb(q)\bigr)^{\frac{1+\gamma-\delta}{2}}   \af\bigl(F(.,q)\bigr) R_0\bigl(\bar{z}\bigr)^{\delta}\Psi \Bigr\| \Bigl\|  R_0\bigl(z- \wb(q)\bigr)^{\frac{1+\delta-\gamma}{2}}  \af\bigl(G(.,q)\bigr) R_0\bigl(z\bigr)^{\gamma} \Phi \Bigr\| d q,
 \end{align*}
 and using Lemma~\ref{RegTermAlone}:
 \begin{align*}
 &\int  \Bigl| \Bigl\langle \Psi \Big| R_0\bigl(z\bigr)^{\delta} \cf\bigl(F(.,q)\bigr)  R_0\bigl(z- \wb(q)\bigr)  \af\bigl(G(.,q)\bigr) R_0\bigl(z\bigr)^{\gamma} \Phi \Bigr\rangle  \Bigr| d q \\
 & \quad\leq  \int \biggl\{  \biggl(\int \frac{|F(k_1,q)|^2}{\wf(k_1)\bigl[\wf(k_1)+\wb(q)+|\Re(z)|\bigr]^{2(\delta+\gamma)-1}} dk_1\biggr)^{\frac12}\\
 & \qquad + \biggl(\int \frac{|F(k_1,q)|^2}{\bigl[\wb(q)+|\Re(z)|\bigr]^{2\delta}\bigl[\wb(q)+\wf(k_1)+|\Re(z)|\bigr]^{2\gamma}}   dk_1\biggr)^{\frac12} \biggr\}  \Bigl\|  R_0\bigl(z- \wb(q)\bigr)^{\frac{1+\delta-\gamma}{2}}  \af\bigl(G(.,q)\bigr) R_0\bigl(z\bigr)^{\gamma} \Phi \Bigr\| d q \\
 &\quad \leq \int \biggl\{  \biggl(\int \frac{|F(k_1,q)|^2}{\wf(k)\bigl[\wf(k_1)+\wb(q)+|\Re(z)|\bigr]^{2(\delta+\gamma)-1}} dk_1\biggr)^{\frac12}\\
 &\qquad + \biggl(\int \frac{|F(k_1,q)|^2}{\bigl[\wb(q)+|\Re(z)|\bigr]^{2\delta}\bigl[\wb(q)+\wf(k_1)+|\Re(z)|\bigr]^{2\gamma}}   dk_1\biggr)^{\frac12} \biggr\}   \\
 &\qquad \quad \biggl\{  \biggl(\int \frac{|G(k_2,q)|^2}{\wf(k_2)\bigl[\wf(k_2)+\wb(q)+|\Re(z)|\bigr]^{2(\delta+\gamma)-1}} dk_2\biggr)^{\frac12}\\
 &\qquad +   \biggl(\int \frac{|G(k_2,q)|^2}{\bigl[\wb(q)+|\Re(z)|\bigr]^{2\delta}\bigl[\wb(q)+\wf(k_2)+|\Re(z)|\bigr]^{2\gamma}}   dk_2\biggr)^{\frac12}\biggr\} dq\\
 &\quad \leq  \biggl\{\biggl(\int \frac{|F(k_1,q)|^2}{\wf(k_1)\bigl[\wf(k_1)+\wb(q)+|\Re(z)|\bigr]^{2(\delta+\gamma)-1}} dk_1 dq\biggr)^{\frac12}\\
& \qquad + \biggl(\int \frac{|F(k_1,q)|^2}{\bigl[\wb(q)+|\Re(z)|\bigr]^{2\delta}\bigl[\wb(q)+\wf(k_1)+|\Re(z)|\bigr]^{2\gamma}}dk_1 dq\biggr)^{\frac12} \biggr\}\\
&\qquad \quad \biggl\{ \biggl(\int \frac{|G(k_2,q)|^2}{\wf(k_2)\bigl[\wf(k_2)+\wb(q)+|\Re(z)|\bigr]^{2(\delta+\gamma)-1}} dk_2 dq \biggr)^{\frac12}\\
&\qquad \quad \biggl(\int \frac{|G(k_2,q)|^2}{\bigl[\wb(q)+|\Re(z)|\bigr]^{2\delta}\bigl[\wb(q)+\wf(k_2)+|\Re(z)|\bigr]^{2\gamma}}   dk_2 dq \biggr)^{\frac12}\biggr\}.
\end{align*}

In the same way as before, Lemma \ref{EstimateWtihRespectToOtherPowers} can now be used  to conclude that 
 \begin{align*}
 &\int  \Bigl| \Bigl\langle \Psi \Big| R_0\bigl(z\bigr)^{\delta} \cf\bigl(F(.,q)\bigr)  R_0\bigl(z- \wb(q)\bigr)  \af\bigl(G(.,q)\bigr) R_0\bigl(z\bigr)^{\gamma} \Phi \Bigr\rangle  \Bigr| d q \\
 &\quad \leq  \biggl\{C_s\biggl(\int \frac{|F(k_1,q)|^2}{\wb(q)\bigl[\wb(q)+|\Re(z)|\bigr]^{2(\delta+\gamma)-1}} dk_1 dq\biggr)^{\frac12}\\
& \qquad + \biggl(\int \frac{|F(k_1,q)|^2}{\wb(q)\bigl[\wb(q)|\Re(z)|\bigr]^{2(\gamma+\delta)-1}}dk_1 dq\biggr)^{\frac12} \biggr\}\\
&\qquad \quad \biggl\{ C_s\biggl(\int \frac{|G(k_2,q)|^2}{\wb(q)\bigl[\wb(q)+|\Re(z)|\bigr]^{2(\delta+\gamma)-1}} dk_2 dq \biggr)^{\frac12}\\
&\qquad \quad \biggl(\int \frac{|G(k_2,q)|^2}{\wb(q)\bigl[\wb(q)+\wf(k_2)+|\Re(z)|\bigr]^{2(\gamma+\delta)-1}}   dk_2 dq \biggr)^{\frac12}\biggr\}.
\end{align*}
\end{proof}

\begin{lem}
\label{IE5}
Let us assume that $\Re(z)<-1$ and $0 \leq \gamma \leq 1$. We then have: 
\begin{align*}
\Bigl\|  H^{\ab \af}(F_1) R_0\bigl(z\bigr)  H^{\ab \cf}(F_2) R_0\bigl(z\bigr) H^{\cb \cf} (F_3) R_0\bigl(z\bigr)^{\gamma} \Bigr\| & \leq \left(64 C_s^{4+2\gamma}+ 3C_s^{\frac{4+2\gamma}{3}} + 2 C_s\right)  K^{(3)}_{z, \gamma}(F_1,F_2,F_3)\\
 \Bigl\| R_0\bigl(z\bigr)^{\gamma} H^{\ab \af}(F_3) R_0\bigl(z\bigr)  H^{\cb \af}(F_2) R_0\bigl(z\bigr) H^{\cb \cf} (F_1)  \Bigr\| & \leq \left(64 C_s^{4+2\gamma}+ 3C_s^{\frac{4+2\gamma}{3}} + 2 C_s\right) K^{(3)}_{z, \gamma}(F_1,F_2,F_3)
\end{align*}
where 
\begin{equation}
K^{(3)}_{z,\gamma}(F_1,F_2,F_3) = \prod^{3}_{i=1}\biggl(\int \frac{|F_i(k_i,q_i)|^2}{\wf(k_i) \bigl[\wf(k_i)+|\Re(z)|\bigr]^{\frac{1}{3} + 2\frac{\gamma}{3}}}d k_1 dq_1\biggr)^{\frac12}.
\end{equation}
\end{lem}

\begin{proof} Let $\Psi,\Phi\in \mathscr H$. We start by reordering the bosonic operators:
\begin{align}\label{C5-Step1}
 \nonumber  & \Bigl|\Bigl\langle \Psi \Big| H^{\ab \af}(F_1) R_0\bigl(z\bigr)  H^{\ab \cf}(F_2) R_0\bigl(z\bigr) H^{\cb \cf}(F_3) R_0\bigl(z\bigr)^{\gamma}   \Phi \Bigr\rangle \Bigr| \\
\nonumber   & \quad  \leq  \biggl|\biggl\langle\Psi \bigg|  \int \af\bigl(F_1(.,q_1)\bigr) \cb(q_3) R_0\bigl(z  - \wb(q_1)-\wb(q_3)\bigr)  \cf\bigl(F_2(.,q_2)\bigr) \\
\nonumber &\qquad \quad \ab(q_1) R_0\bigl(z-\wb(q_3) - \wb(q_2)\bigr)  \cf\bigl(F_3(.,q_3)\bigr)\ab(q_2) dq_1 dq_2 dq_3  R_0\bigl(z\bigr)^{\gamma}\Phi \biggr>\biggr|\\
\nonumber & \qquad + \biggl|\biggl\langle \Psi \bigg|  \int \af\bigl(F_1(.,q)\bigr)  R_0\bigl(z  - \wb(q)\bigr)  \cf\bigl(F_2(.,q_2)\bigr) R_0\bigl(z-\wb(q) - \wb(q_2)\bigr)  \\
\nonumber &\qquad \quad  \cf\bigl(F_3(.,q)\bigr)\ab(q_2)  dq_2 dq R_0\bigl(z\bigr)^{\gamma}\Phi \biggr\rangle\biggr|\\
\nonumber & \qquad  + \biggl|\biggl\langle \Psi \bigg|  \int \af\bigl(F_1(.,q_1)\bigr)  R_0\bigl(z  - \wb(q_1)\bigr)  \cf\bigl(F_2(.,q)\bigr)R_0\bigl(z-\wb(q) - \wb(q_1)\bigr)  \\ 
&\qquad \quad \cf\bigl(F_3(.,q)\bigr)  
\ab(q_1) dq_1 dq R_0\bigl(z\bigr)^{\gamma}\Phi \biggr\rangle\biggr|.
\end{align}
We now normal order the fermionic operators related to contracted bosonic operators in the two last terms on the right-hand side of \eqref{C5-Step1}:
\begin{align}\label{C5-Step2a}
& \nonumber \biggl|\biggl\langle \Psi \bigg|  \int \af\bigl(F_1(.,q)\bigr)  R_0\bigl(z  - \wb(q)\bigr)  \cf\bigl(F_2(.,q_2)\bigr) R_0\bigl(z-\wb(q) - \wb(q_2)\bigr) \cf\bigl(F_3(.,q)\bigr)\ab(q_2)  dq_2 dq R_0\bigl(z\bigr)^{\gamma}\Phi \biggr\rangle\biggr| \\
&\nonumber \quad \leq
\biggl|\biggl\langle \Psi \bigg|  \int \overline{F_1(k,q)} F_3(k,q)  R_0\bigl(z - \wb(q)-\wf(k)\bigr)  \cf\bigl(F_2(.,q_2)\bigr)\\ 
&\nonumber \qquad \quad
 R_0\bigl(z-\wb(q) - \wb(q_2)-\wf(k)\bigr)  \ab(q_2) dq dq_2 dk  R_0\bigl(z\bigr)^{\gamma} \Phi \biggr\rangle\biggr| \\
&\nonumber \qquad  +
\biggl|\biggl\langle \Psi \bigg|  \int \overline{F_1(k_1,q)} F_3(k_3,q) \cf(k_3) R_0\bigl(z - \wb(q)-\wf(k_1)-\wf(k_3)\bigr)  \cf\bigl(F_2(.,q_2)\bigr) \\ 
& \nonumber \qquad \quad
 \af(k_1) R_0\bigl(z-\wb(q) - \wb(q_2)-\wf(k_3)\bigr)  \ab(q_2) dq dq_2 dk_1 dk_3 R_0\bigl(z\bigr)^{\gamma} \Phi \biggr\rangle\biggr| \\ 
&\nonumber \qquad  +
\biggl|\biggl\langle \Psi \bigg| \int \overline{F_1(k,q) } F_3(k_3,q) F_2(k,q_2)  R_0\bigl(z - \wb(q)-\wf(k) \bigr) \cf(k_3) \\ 
&\qquad \quad  R_0\bigl(z-\wb(q) - \wb(q_2)-\wf(k_3)\bigr)  \ab(q_2) dq dq_2 dk dk_3 R_0\bigl(z\bigr)^{\gamma}\Phi \biggr\rangle\biggr|
\end{align}
and
\begin{align}\label{C5-Step2b}
\nonumber & \biggl|\biggl\langle \Psi \bigg|  \int \af\bigl(F_1(.,q_1)\bigr)  R_0\bigl(z  - \wb(q_1)\bigr)  \cf\bigl(F_2(.,q)\bigr)R_0\bigl(z-\wb(q) - \wb(q_1)\bigr)  
\cf\bigl(F_3(.,q)\bigr)  
\ab(q_1) dq_1 dq R_0(z)^{\gamma}\Phi \biggr\rangle\biggr|\\
&\nonumber \quad \leq \biggl|\biggl\langle \Psi \bigg| \int F_2(k_2,q) F_3(k_3,q)  \cf(k_2)\af\bigl(F_1(.,q_1)\bigr) R_0\bigl(z  - \wb(q_1)-\wf(k_2)\bigr)   \cf(k_3)\\
 &\nonumber \qquad \quad R_0\bigl(z-\wb(q)-\wf(k_3)- \wb(q_1)\bigr)  \ab(q_1) dq_1 dq dk_2 dk_3  R_0\bigl(z\bigr)^{\gamma} \Phi \biggr\rangle\biggr| \\ 
&\nonumber \qquad + \biggl|\biggl\langle \Psi \bigg|  \int \overline{F_1(k, q_1)} F_2(k,q) F_3(k_3,q)   \cf(k_3) R_0\bigl(z-\wf(k)-\wb(q_1)-\wf(k_3)\bigr) \\ 
&\qquad \quad R_0\bigl(z-\wb(q)-\wf(k_3)-\wb(q_1)\bigr)\ab(q_1)  dq_1 dq d k d k_3  R_0\bigl(z\bigr)^{\gamma} \Phi \biggr\rangle\biggr|.
\end{align}

We start by estimating the first term on the right-hand side of \eqref{C5-Step1}
\begin{align}\label{C5-Step3a}
   &\nonumber \biggl|\biggl\langle \Psi \bigg|  \int \af\bigl(F_1(.,q_1)\bigr) \cb(q_3) R_0\bigl(z  - \wb(q_1)-\wb(q_3)\bigr)  \cf\bigl(F_2(.,q_2)\bigr)   \\
& \nonumber \qquad \quad \ab(q_1) R_0\bigl(z-\wb(q_3) - \wb(q_2)\bigr)  \cf\bigl(F_3(.,q_3)\bigr)\ab(q_2) dq_1 dq_2 dq_3 R_0\bigl(z\bigr)^{\gamma}\Phi \biggr\rangle \biggr|\\
& \nonumber \quad \leq \int \Bigl|\Bigl\langle R_0\bigl(\bar{z}  - \wb(q_1)-\wb(q_3)\bigr)^{\frac23+\frac13\gamma} \cf\bigl(F_1(.,q_1)\bigr) \ab(q_3)\Psi \Big|  \\
&\nonumber  \qquad \quad      R_0\bigl(z  - \wb(q_1)-\wb(q_3)\bigr)^{\frac13-\frac13\gamma}  \cf\bigl(F_2(.,q_2)\bigr)\ab(q_1) R_0\bigl(z-\wb(q_3) - \wb(q_2)\bigr)^{\frac13+\frac23\gamma}\\
&\nonumber  \qquad  \quad R_0\bigl(z-\wb(q_3) - \wb(q_2)\bigr)^{\frac23-\frac23\gamma}  \cf\bigl(F_3(.,q_3)\bigr)\ab(q_2)  R_0\bigl(z\bigr)^{\gamma}\Phi \Bigr\rangle\Bigr|dq_1 dq_2 dq_3\\
& \nonumber \quad \leq \int \Bigl\| R_0\bigl(\bar{z}  - \wb(q_1)-\wb(q_3)\bigr)^{\frac23+\frac13\gamma} \cf\bigl(F_1(.,q_1)\bigr) \ab(q_3)\Psi \Bigr\| \\
& \nonumber \qquad \quad \Bigl\|    R_0\bigl(z  - \wb(q_1)-\wb(q_3)\bigr)^{\frac13-\frac13\gamma}  \cf\bigl(F_2(.,q_2)\bigr)\ab(q_1) R_0\bigl(z-\wb(q_3) - \wb(q_2)\bigr)^{\frac13+\frac23\gamma}\\
&  \qquad \quad  R_0\bigl(z-\wb(q_3) - \wb(q_2)\bigr)^{\frac23-\frac23\gamma}  \cf\bigl(F_3(.,q_3)\bigr)\ab(q_2)  R_0\bigl(z\bigr)^{\gamma}\Phi\Bigr\| dq_1 dq_2 dq_3,
\end{align}
where we used the Cauchy-Schwarz inequality in the last step. 
Multiplying and dividing by $\wb(q_3)^{\frac12} [\wb(q_3) + |\Re(z)|]^{\frac16+\frac13\gamma}$, followed by an application of Cauchy-Schwarz inequality with respect to the $q_3$-integration yields:
\begin{align}\label{C5-Step3b}
&\nonumber \int \Bigl\| R_0\bigl(\bar{z}  - \wb(q_1)-\wb(q_3)\bigr)^{\frac23+\frac13\gamma} \cf\bigl(F_1(.,q_1)\bigr) \ab(q_3)\Psi \Bigr\| \\
&\nonumber  \qquad \quad \Bigl\|    R_0\bigl(z  - \wb(q_1)-\wb(q_3)\bigr)^{\frac13-\frac13\gamma}  \cf\bigl(F_2(.,q_2)\bigr)\ab(q_1) R_0\bigl(z-\wb(q_3) - \wb(q_2)\bigr)^{\frac13+\frac23\gamma}\\
& \nonumber \qquad \quad  R_0\bigl(z-\wb(q_3) - \wb(q_2)\bigr)^{\frac23-\frac23\gamma}  \cf\bigl(F_3(.,q_3)\bigr)\ab(q_2)  R_0\bigl(z\bigr)^{\gamma}\Phi\Bigr\| dq_1 dq_2 dq_3\\
& \nonumber \quad \leq \int \biggl(\int \wb(q_3) \bigl[\wb(q_3) + |\Re(z)|\bigr]^{\frac13+\frac23\gamma}\biggl\| R_0\bigl(-|\Re(z)|  -\wb(q_3)\bigr)^{\frac23+\frac13\gamma} \cf\bigl(F_1(.,q_1)\bigr) \ab(q_3)\Psi \biggr\|^2 dq_3\biggr)^{\frac12} \\
& \nonumber\qquad  \quad\biggl(\int \frac{\biggl\|    R_0\bigl(z  - \wb(q_1)-\wb(q_3)\bigr)^{\frac13-\frac13\gamma}  \cf\bigl(F_2(.,q_2)\bigr)\ab(q_1) R_0\bigl(z-\wb(q_3) - \wb(q_2)\bigr)^{\frac13+\frac23\gamma}}{\wb(q_3) \bigl[\wb(q_3) + |\Re(z)|\bigr]^{\frac13+\frac23\gamma}} \\
& \qquad \quad  R_0\bigl(z-\wb(q_3) - \wb(q_2)\bigr)^{\frac23-\frac23\gamma}  \cf\bigl(F_3(.,q_3)\bigr)\ab(q_2)  R_0\bigl(z\bigr)^{\gamma}\Phi\biggr\|^2dq_3\biggr)^{\frac12} dq_1 dq_2.
\end{align}
 Lemma~\ref{EstiamtesofAsharp} can be used to establish that 
\begin{equation*}
    \int \wb(q_3) \bigl[\wb(q_3) + |\Re(z)|\bigr]^{\frac13+\frac23\gamma}\Bigl\| R_0\bigl(-|\Re(z)|  -\wb(q_3)\bigr)^{\frac23+\frac13\gamma} \cf\bigl(F_1(.,q_1)\bigr) \ab(q_3)\Psi \Bigr\|^2 dq_3 \leq 4 \bigl\|F_1(.,q_1)\bigr\|^2 \bigl\|\Psi\bigr\|^2.
\end{equation*}
Inserting this into \eqref{C5-Step3b}, we may continue the estimate  \eqref{C5-Step3a}, multiplying and dividing by $\wb(q_1)^{\frac12} [\wb(q_1)+|\Re(z)|]^{\frac16 + \frac13\gamma}$: 
\begin{align*}
   &\biggl|\biggl\langle \Psi \bigg|  \int \af\bigl(F_1(.,q_1)\bigr) \cb(q_3) R_0\bigl(z  - \wb(q_1)-\wb(q_3)\bigr)  \cf\bigl(F_2(.,q_2)\bigr) \\
& \qquad\quad  \ab(q_1) R_0\bigl(z-\wb(q_3) - \wb(q_2)\bigr)  \cf\bigl(F_3(.,q_3)\bigr)\ab(q_2) dq_1 dq_2 dq_3 R_0\bigl(z\bigr)^{\gamma}\Phi \biggr\rangle\biggr|\\
& \quad \leq 4 \bigl\|\Psi\bigr\|  \int \frac{\|F_1(.,q_1)\|}{\wb(q_1)^{\frac12} \bigl[\wb(q_1)+|\Re(z)|\bigr]^{\frac16 + \frac13\gamma}} \biggl(\frac{1}{\wb(q_3) \bigl[\wb(q_3) + |\Re(z)|\bigr]^{\frac13+\frac23\gamma}} \int \wb(q_1) \bigl[\wb(q_1)+|\Re(z)|\bigr]^{\frac13 + \frac23\gamma} \\
& \qquad \quad  \Bigl\|    R_0\bigl(z  - \wb(q_1)-\wb(q_3)\bigr)^{\frac13-\frac13\gamma}  \cf\bigl(F_2(.,q_2)\bigr)\ab(q_1) R_0\bigl(z-\wb(q_3) - \wb(q_2)\bigr)^{\frac13+\frac23\gamma} \\
& \qquad \quad  R_0\bigl(z-\wb(q_3) - \wb(q_2)\bigr)^{\frac23-\frac23\gamma}  \cf\bigl(F_3(.,q_3)\bigr)\ab(q_2)  R_0\bigl(z\bigr)^{\gamma}\Phi\Bigr\|^2dq_3\biggr)^{\frac12} dq_1 dq_2. \\
\end{align*}
Using Cauchy-Schwarz inequality with respect to the $q_1$-integration, followed by another application of Lemma~\ref{EstiamtesofAsharp}, also with respect to the $q_1$-integration,  we arrive at
\begin{align*}
   &\biggl|\biggl\langle \Psi \bigg|  \int \af\bigl(F_1(.,q_1)\bigr) \cb(q_3) R_0\bigl(z  - \wb(q_1)-\wb(q_3)\bigr)  \cf\bigl(F_2(.,q_2)\bigr)  \\
& \qquad \quad \ab(q_1) R_0\bigl(z-\wb(q_3) - \wb(q_2)\bigr)  \cf\bigl(F_3(.,q_3)\bigr)\ab(q_2) dq_1 dq_2 dq_3 R_0\bigl(z\bigr)^{\gamma}\Phi \biggr\rangle\biggr|\\
& \quad \leq 4 \bigl\|\Psi\bigr\|  \biggl(\int \frac{\|F_1(.,q'_1)\|^2}{\wb(q'_1) \bigl[\wb(q'_1)+|\Re(z)|\bigr]^{\frac13 + \frac23\gamma}} dq'_1\biggr)^{\frac12} \int \biggl( \int \frac{\wb(q_1) \bigl[\wb(q_1)+|\Re(z)|\bigr]^{\frac13 + \frac23\gamma}}{\wb(q_3) \bigl[\wb(q_3) + |\Re(z)|\bigr]^{\frac13+\frac23\gamma}} \\
&\qquad \quad \Bigl\|    R_0\bigl(z  - \wb(q_1)-\wb(q_3)\bigr)^{\frac13-\frac13\gamma}  \cf\bigl(F_2(.,q_2)\bigr)\ab(q_1) R_0\bigl(z-\wb(q_3) - \wb(q_2)\bigr)^{\frac13+\frac23\gamma} \\
& \qquad \quad  R_0\bigl(z-\wb(q_3) - \wb(q_2)\bigr)^{\frac23-\frac23\gamma}  \cf\bigl(F_3(.,q_3)\bigr)\ab(q_2)  R_0\bigl(z\bigr)^{\gamma}\Phi\Bigr\|^2dq_3 dq_1\biggr)^{\frac12} dq_2 \\
& \quad \leq 16 \bigl\|\Psi\bigr\|  \biggl(\int \frac{\|F_1(.,q'_1)\|^2}{\wb(q'_1) \bigl[\wb(q'_1)+|\Re(z)|\bigr]^{\frac13 + \frac23\gamma}} dq'_1\biggr)^{\frac12} \int \bigl\|F_2(.,q_2)\bigr\| \biggl(\int\frac{1}{\wb(q_3) \bigl[\wb(q_3) + |\Re(z)|\bigr]^{\frac13+\frac23\gamma}}    \\
& \qquad \quad  \Bigl\| R_0\bigl(z-\wb(q_3) - \wb(q_2)\bigr)^{\frac23-\frac23\gamma}  \cf\bigl(F_3(.,q_3)\bigr)\ab(q_2)  R_0\bigl(z\bigr)^{\gamma}\Phi\Bigr\|^2dq_3\biggr)^{\frac12} dq_2  \\
\end{align*}
Continuing a third time, multiplying and dividing by $\wb(q_2)^{\frac12} [\wb(q_2) + |\Re(z)|]^{\frac16+\frac13\gamma}$ followed by an application of the Cauchy-Schwarz inequality with respect to the $q_2$-integration and yet another application of  Lemma~\ref{EstiamtesofAsharp}, also with respect to the $q_2$-integration, yields
\begin{align*}
   &\biggl|\biggl\langle\Psi \bigg|  \int \af\bigl(F_1(.,q_1)\bigr) \cb(q_3) R_0\bigl(z  - \wb(q_1)-\wb(q_3)\bigr)  \cf\bigl(F_2(.,q_2)\bigr) \\
& \qquad \quad \ab(q_1) R_0\bigl(z-\wb(q_3) - \wb(q_2)\bigr)  \cf\bigl(F_3(.,q_3)\bigr)\ab(q_2) dq_1 dq_2 dq_3 R_0\bigl(z\bigr)^{\gamma}\Phi \biggr\rangle\biggr|\\
& \quad \leq 16 \bigl\|\Psi\bigr\|  \biggl(\int \frac{\|F_1(.,q_1)\|^2}{\wb(q_1) \bigl[\wb(q_1)+|\Re(z)|\bigr]^{\frac13 + \frac23\gamma}} dq_1\biggr)^{\frac12} \int \frac{\|F_2(.,q_2)\|}{\wb(q_2)^{\frac12} \bigl[\wb(q_2) + |\Re(z)|\bigr]^{\frac16+\frac13\gamma}}   \\
& \qquad  \quad \biggl(\int\frac{\wb(q_2) \bigl[\wb(q_2) + |\Re(z)|\bigr]^{\frac13+\frac23\gamma}}{\wb(q_3) \bigl[\wb(q_3) + |\Re(z)|\bigr]^{\frac13+\frac23\gamma}}   \biggl\| R_0\bigl(z-\wb(q_3) - \wb(q_2)\bigr)^{\frac23-\frac23\gamma}  \cf\bigl(F_3(.,q_3)\bigr)\ab(q_2)  R_0\bigl(z\bigr)^{\gamma}\Phi\biggr\|^2dq_3 \biggr)^{\frac12} dq_2  \\
& \quad \leq 16 \bigl\|\Psi\bigr\|  \biggl(\int \frac{\|F_1(.,q_1)\|^2}{\wb(q_1) \bigl[\wb(q_1)+|\Re(z)|\bigr]^{\frac13 + \frac23\gamma}} dq_1\biggr)^{\frac12} \biggl(\int \frac{\|F_2(.,q'_2)\|^2}{\wb(q'_2) \bigl[\wb(q'_2) + |\Re(z)|\bigr]^{\frac13+\frac23\gamma}} dq'_2\biggr)^{\frac12}   \\
& \qquad   \quad  \biggl(\int\frac{\wb(q_2) \bigl[\wb(q_2) + |\Re(z)|\bigr]^{\frac13+\frac23\gamma}}{\wb(q_3) \bigl[\wb(q_3) + |\Re(z)|\bigr]^{\frac13+\frac23\gamma}}   \Bigl\|R_0\bigl(z-\wb(q_3) - \wb(q_2)\bigr)^{\frac23-\frac23\gamma}  \cf\bigl(F_3(.,q_3)\bigr)\ab(q_2)  R_0\bigl(z\bigr)^{\gamma}\Phi\Bigr\|^2dq_3  dq_2 \biggr)^{\frac12}\\
& \quad  \qquad 64 \bigl\|\Phi\bigr\| \bigl\|\Psi\bigr\|  \biggl(\int \frac{\|F_1(.,q_1)\|^2}{\wb(q_1) \bigl[\wb(q_1)+|\Re(z)|\bigr]^{\frac13 + \frac23\gamma}} dq_1\biggr)^{\frac12} \biggl(\int \frac{\|F_2(.,q_2)\|^2}{\wb(q_2) \bigl[\wb(q_2) + |\Re(z)|\bigr]^{\frac13+\frac23\gamma}} dq_2\biggr)^{\frac12}   \\
& \qquad  \quad \biggl(\int \frac{\|F_3(.,q_3)\|^2}{\wb(q_3) \bigl[\wb(q_3) + |\Re(z)|\bigr]^{\frac13+\frac23\gamma}} dq_3\biggr)^{\frac12}\\
& \quad \leq 64 C_s^{4+2\gamma} K^{(3)}_{z,\gamma}(F_1,F_2,F_3).
\end{align*}
Here Lemma~\ref{EstimateWtihRespectToOtherPowers} have been used in the last inequality. This proof can be straightforwardly modified so that one can estimate the terms on the right-hand sides of \eqref{C5-Step2a} and \eqref{C5-Step2b}, using Lemma~\ref{EstiamtesofAsharp},~\ref{EstiamtesofAsharp1} and Lemma~\ref{EstimateWtihRespectToOtherPowers}. 

\end{proof}

The same type of arguments enable us to prove the following lemma

\begin{lem}
\label{IE6}
Let us assume that $\Re(z)<-1$ and $0\leq \gamma \leq 1$ we then have: 
\begin{align*}
\Bigl\|  H^{\ab \cf}(F_1) R_0(z)  H^{\ab \cf}(F_2) R_0(z) H^{\cb \cf} (F_3) R_0(z)^{\gamma} \Bigr\| & \leq \Bigl(64 C_s^{4+2\gamma}+2C_s^{\frac{4+2\gamma}{3}}\Bigr)  K^{(3)}_{z, \gamma}(F_1,F_2,F_3)\\
 \Bigl\| R_0(z)^{\gamma} H^{\ab \af}(F_3) R_0(z)  H^{\cb \af}(F_2) R_0(z) H^{\cb \af} (F_1)  \Bigr\| & \leq \Bigl(64 C_s^{4+2\gamma}+2C_s^{\frac{4+2\gamma}{3}}\Bigr)  K^{(3)}_{z, \gamma}(F_1,F_2,F_3).
\end{align*}
where 
\begin{equation}
K^{(3)}_{z,\gamma}(F_1,F_2,F_3) = \prod^{3}_{i=1}\biggl(\int \frac{|F_i(k_i,q_i)|^2}{\wf(k_i) \bigl[\wf(k_i)+|\Re(z)|\bigr]^{\frac{1}{3} + 2\frac{\gamma}{3}}}d k_i dq_i\biggr)^{\frac12}.
\end{equation}
\end{lem}

\begin{lem}
\label{RegUniformKernel}
Let $\beta \in [0,1]$. The following holds true:
\begin{enumerate}
    \item If $\beta  > \frac{d}{2}-p$, then $ K^{(1)}_{z,\beta}(G^{(\sharp)}_{\Lambda})$ is uniformly bounded with respect to $\Lambda$. Moreover, for any $\epsilon>0$ there exists $R$ such that for any $\Lambda, \Lambda'\geq R$, 
    \[
    K^{(1)}_{z,\beta}\bigl(G^{(\sharp)}_{\Lambda} - G^{(\sharp)}_{\Lambda'}\bigr)\leq \epsilon.
    \]
   \item If $\beta  > d-2p-1$, then $  K^{(2)}_{z,\beta}(G^{(\sharp)}_{\Lambda},G^{(\sharp)}_{\Lambda})$ is uniformly bounded with respect to $\Lambda$. Moreover, for any $\epsilon>0$ there exists $R$ such that for any $\Lambda, \Lambda'\geq R$, 
   \[
    K^{(2)}_{z,\beta}\bigl(G^{(\sharp)}_{\Lambda} - G^{(\sharp)}_{\Lambda'},G^{(\sharp)}_{\Lambda}\bigr) \leq \epsilon\qquad \textup{and} \qquad 
         K^{(2)}_{z,\beta}\bigl(G^{(\sharp)}_{\Lambda'}, G^{(\sharp)}_{\Lambda} - G^{(\sharp)}_{\Lambda'}\bigr) \leq \epsilon.
   \]
   \item If $\beta  > \frac{3}{2}d-3p-2$, then $  K^{(3)}_{z,\beta}(G^{(\sharp)}_{\Lambda},G^{(\sharp)}_{\Lambda},G^{(\sharp)}_{\Lambda})$ is uniformly bounded with respect to $\Lambda$. Moreover, for any $\epsilon>0$ there exists $R$ such that for any $\Lambda, \Lambda'\geq R$, 
   \[
      K^{(3)}_{z,\beta}\bigl(G^{(\sharp)}_{\Lambda} - G^{(\sharp)}_{\Lambda'},G^{(\sharp)}_{\Lambda},G^{(\sharp)}_{\Lambda}\bigr) \leq \epsilon,\qquad
       K^{(3)}_{z,\beta}\bigl(G^{(\sharp)}_{\Lambda'}, G^{(\sharp)}_{\Lambda} - G^{(\sharp)}_{\Lambda'},G^{(\sharp)}_{\Lambda}\bigr) \leq \epsilon\qquad \textup{and}\qquad
       K^{(3)}_{z,\beta}\bigl(G^{(\sharp)}_{\Lambda'}, G^{(\sharp)}_{\Lambda'}, G^{(\sharp)}_{\Lambda} - G^{(\sharp)}_{\Lambda'}\bigr) \leq \epsilon.
   \]
\end{enumerate}
\end{lem}
\begin{proof}
We prove only the Lemma for $K^{(1)}_{z,\beta}(G^{(j)}_{\Lambda})$, $j=1,2$, the other statements can be proved in the same way. For $j=1,2$, we have
\begin{align*}
    K^{(1)}_{z,\beta}\bigl(G^{(j)}_{\Lambda}\bigr)^2 &= \int \frac{\bigl|G^{(j)}_{\Lambda}(k,q)\bigr|^2}{\bigl[\wf(k)+|\Re(z)|\bigr]^{2\beta-1} \wf(k)} dq dk \\
    & = \int \frac{\bigl|h^{(j)}(k,q)\bigr|\bigl|g(k- (-1)^j q)\bigr|^{2}  \chi_\Lambda(k)^2 \chi_\Lambda(q)^2}{\wb(q)^{2p}\bigl[\wf(k)+|\Re(z)|\bigr]^{2\beta-1} \wf(k)} dq dk. 
\end{align*}
By using the same type of argument as the one used in the proof of Lemma~\ref{EstimateWtihRespectToOtherPowers}, there exists a constant $C$, only depending on the masses $\mb$ and $\mf$, as well as the exponent $\beta$, such that we have: 
\begin{align*}
    K^{(1)}_{z,\beta}\bigl(G^{(j)}_{\Lambda}\bigr) & \leq C \bigl\|g\bigr\| \bigl\|h^{(j)}\bigr\|_{\infty} \biggl(\int \frac{1}{\bigl[\wf(k)+|\Re(z)|\bigr]^{2\beta-1} \wf(k)^{1+2p}}  dk \biggr)^{\frac12}.
\end{align*}
The claim now follows, since the integral is finite for $\beta>\frac{d}2 - p$.
\end{proof}

\section{Construction Theorem}

We recall, without proof, a construction theorem presented in \cite{AW2017}. 
\begin{Thannex}
\label{ThAW2017}
Let $\mathscr{H}$ be a Hilbert space, $\delta>0$, and $\rho_{\delta} = \{z \in \mathbb{C} \, |\, \Re(z) < -\delta \}$. Let $\{R(z)\}_{z \in \rho_{\delta}}$ be a family of bounded linear operators in $\mathscr{H}$ with: 
\begin{enumerate}
\item $R(z)^* = R(\overline{z})$~~for all $z\in \rho_{\delta} $,
\item $R(z_1) - R(z_2) = (z_2-z_1)R(z_1) R(z_2)$~~for all $z\in \rho_{\delta} $,
\item $z R(z) \Psi \to \Psi $ as $\Re(z) \to - \infty$  for all $\Psi \in \mathscr{H}$.
\end{enumerate}
Then, there exists a self-adjoint operator $H \colon D(H) \subset \mathscr{H} \to \mathscr{H}$ with  $\sigma(H) \subset [-\delta, \infty)$ and $(z-H)^{-1} = R(z)$ for all $z \in \rho_{\delta}$. 
\end{Thannex}

\noindent\textbf{Acknowledgement.} The authors are grateful for support by the Independent Research Fund Denmark, via the project grant “Mathematical Aspects of Ultraviolet Renormalization” (8021-00242B).

\bibliographystyle{amsalpha}

\end{document}